\documentclass[lettersize,journal]{IEEEtran}
\usepackage[IEEEtran]{./research17}

\usepackage{multirow}    
\usepackage{amsmath,amssymb,amsfonts}
\usepackage{graphicx}
\usepackage{textcomp}
\usepackage{xcolor}
\def\BibTeX{{\rm B\kern-.05em{\sc i\kern-.025em b}\kern-.08em
		T\kern-.1667em\lower.7ex\hbox{E}\kern-.125emX}}

\usepackage[numbers,sort&compress]{natbib}
\usepackage{algorithm}
\usepackage{algorithmic}

\usepackage{fancyhdr}

\usepackage{graphicx}
\usepackage{subcaption}

\newtheorem{thm}{Theorem}
\newtheorem{lem}{Lemma}
\newtheorem{exam}{Example}
\newtheorem{rmk}{Remark}
\newtheorem{prop}{Proposition}

\makeatletter
\def\thanks#1{\protected@xdef\@thanks{\@thanks
		\protect\footnotetext{#1}}}
\makeatother

\usepackage[colorlinks=true,
linkcolor=blue,
citecolor=blue,
urlcolor=blue]{hyperref}

\begin{document}
\title{Capacity Bounds and High-SNR Characterization for  MIMO-OWC Channels Under Average-Power Constraint}

\author{Sufang~Yang, Liang Xia,
	Longguang~Li,~\IEEEmembership{Member,~IEEE,}
	{Jintao~Wang},~\IEEEmembership{Fellow,~IEEE,} \\
	Tao Jiang, Yuxin Wang, Ya Li, Hongjun He, Qixing Wang, and Guangyi Liu
	\thanks{{The two authors contributes equally: Sufang Yang, Liang Xia.}}
	\thanks{{This work was supported in part by the National Natural Science Foundation of China under Grant No. 62101192, in part by the Mobile Information Networks-National Science and Technology Major Project under Grant No.2025ZD1303200, and in part by the Science and Technology Research and Development Program of China State Railway Group Co., Ltd. under Grant No. L2024G005}. \emph{(Corresponding author: Longguang Li.)}}
	\thanks{Sufang~Yang, Tao Jiang, Yuxin Wang, Ya Li, Hongjun He, Qixing Wang, and Guangyi Liu are with the Future Research Laboratory, China Mobile Research Institute, Beijing 100053, China (e-mail: yangsufang@chinamobile.com).}
	\thanks{Liang Xia is with the School of Information Science and Engineering, Southeast University, Nanjing 210096, China, also with the Future Research Laboratory, China Mobile Research Institute, Beijing 100053, China (e-mail: xialiang@chinamobile.com).}
	\thanks{Longguang Li is with SJTU Paris Elite Institute of Technology, Shanghai Jiao Tong University, Shanghai 200240, China (e-mail: llg9012@sjtu.edu.cn).}
	\thanks{{Jintao~Wang is with the Department of Electronic Engineering, Tsinghua University, Beijing 100084, China (e-mail: wangjintao@tsinghua.edu.cn).}}
}

\date{}

\maketitle
\begin{abstract}
This paper investigates the capacity of multiple-input multiple-output (MIMO) optical wireless communication (OWC) channels under a total average-power constraint.
Since different nonnegative input vectors can be mapped to the same image vector and thus induce the same output distribution, we formulate a nonnegative basis pursuit (NN-BP) problem to identify the minimum-$\ell_1$-norm input vector for each image vector. Based on the NN-BP characterization, we derive an equivalent expression for the channel capacity in terms of the image-vector distribution. We then establish computable lower and upper capacity bounds for both $n_\textnormal{T}\geq n_\textnormal{R}$ and $n_\textnormal{T}< n_\textnormal{R}$ cases, and prove that the proposed bounds are asymptotically tight in the high signal-to-noise ratio (SNR) regime. Numerical results for indoor and outdoor OWC scenarios demonstrate that the proposed bounds improve upon existing ones and close the constant gap in the high-SNR regime.
\end{abstract}

\begin{IEEEkeywords}
Channel capacity, intensity modulation and direct detection, optical communication, average-power constraint, MIMO.
\end{IEEEkeywords}

\vspace{-4mm}
\section{Introduction}

\IEEEPARstart{O}{ptical} wireless communication (OWC) has attracted considerable attention as a promising technology for future networks, owing to its high data rates, large unlicensed bandwidth, immunity to electromagnetic interference, low cost, and flexible deployment \cite{Kaushal2017,Pathak2015,Khalighi2014}. Typical OWC systems include indoor visible light communication (VLC) and outdoor free-space optical communication (FSO) \cite{Karunatilaka2015,Zhu2002}. Unlike radio frequency (RF) systems, OWC systems usually adopt an \emph{intensity-modulation/direct-detection} scheme \cite{Tavan2012,Elzanaty2020}. The transmitter employs light-emitting diodes (LEDs) or laser diodes (LDs) to generate optical signal whose power varies at high speed, while the receiver uses photodetectors (PDs) to measure the incident optical power, from which the transmitted information is recovered. The IM/DD mechanism imposes distinctive constraints on OWC channel inputs. Since the transmitted signal represents optical power, the channel input is required to be real-valued and nonnegative, in contrast to the complex-valued counterpart in RF. Moreover, due to practical considerations such as illumination and safety requirements, the channel input is usually subject to an average-power constraint. These fundamentally alter the feasible input space and make classical capacity results for complex-valued RF channels inapplicable. Therefore, this paper investigates the capacity of OWC channels under nonnegativity and an average-power constraint from an information-theoretic perspective \cite{Shannon1949}.

Recent years have witnessed significant progress in the capacity analysis of single-input single-output (SISO) OWC channels. Existing studies have primarily focused on bounding the channel capacity and identifying its asymptotic behavior in the high and low signal-to-noise ratio (SNR) regimes. For instance, \cite{Lapidoth2009} adopted an exponential distribution as the input, which corresponds to the maxentropic source distribution under the average-power constraint. Combined with the entropy power inequality (EPI), this input choice yields a lower bound on channel capacity. On the other hand, by employing the duality approach with different auxiliary output distributions, upper bounds on channel capacity were also developed. Theoretical analysis showed that the resulting bounds are asymptotically tight in the high-SNR regime, whereas a non-negligible gap persists in the low-SNR regime. To tighten the bound at low SNR, recent studies \cite{Li2025,Jiang2024} constructed novel auxiliary output distributions and established upper bounds that are asymptotically optimal in this regime. Beyond the aforementioned classical bounding techniques, alternative approaches have also been developed for capacity analysis of SISO-OWC channels. Regarding upper bounds, \cite{Hranilovic2004,Farid2010,Wang2013, Chaaban2016,Jiang2016} employed a sphere-packing method to refine the bound. Regarding lower bounds, \cite{Hranilovic2004,Farid2010} explored alternative input distributions, such as continuous exponential distributions and discrete geometric distributions with feasible average power, and then optimized their parameters numerically to obtain lower bounds on channel capacity.

Despite the advances in SISO-OWC channels, fewer results are available for multiple-input multiple-output (MIMO) channels. Existing studies have mainly derived capacity bounds and asymptotic results for several representative configurations. For the case $n_\textnormal{T} < n_\textnormal{R}$, where the number of transmit apertures is smaller than that of receive apertures, \cite{Moser2017-ISIT,Chaaban2018-hsnr} showed that this MIMO channel can be transformed equivalently into a set of parallel SISO channels. This equivalent transformation enables the application of existing SISO capacity results, thereby simplifying the derivation of capacity bounds for this configuration.
For the case $n_\textnormal{T}\geq n_\textnormal{R}$, where the number of transmit apertures is not smaller than that of receive apertures, \cite{Moser2017-ISIT} studied the special case with $n_\textnormal{R}=1$. Similarly, this multiple-input single-output (MISO) channel was reduced to an equivalent SISO channel, for which the high-SNR asymptotic capacity was derived. Subsequently, the general MIMO case was considered in \cite{Chaaban2018-hsnr}. Therein, a lower bound on channel capacity was obtained by applying QR decomposition, and an upper bound was derived based on the subadditivity of entropy. Nevertheless, these bounds remain separated by a non-vanishing gap in the high-SNR regime. 

To the best of the authors' knowledge, a unified framework for characterizing the capacity of MIMO-OWC channels under an average-power constraint remains unavailable. Moreover, the constant gap between the existing upper and lower bounds in the high-SNR regime has not yet been closed. Motivated by these limitations, this paper develops an image-space partition based on nonnegative basis pursuit (NN-BP), which serves as a key tool for deriving tight capacity bounds for the considered MIMO-OWC channels. A closely related geometric approach is the minimum-energy signaling approach in \cite{Li2020}, developed for MIMO-OWC channels under both peak- and average-power constraints. In \cite{Li2020}, the peak-power constraint confines the input space to a bounded hypercube; consequently, the image space becomes a bounded zonotope that is partitioned into bounded parallelepipeds. However, that partition cannot be applied to the setting with only an average-power constraint, where the image space is an unbounded cone. The proposed NN-BP partition is tailored to this unbounded geometry and divides the image cone into unbounded conic regions.

The main contributions are summarized as follows.
\begin{itemize}
	\item \emph{Nonnegative Basis Pursuit:} To address the nonuniqueness of the input induced by the many-to-one mapping from nonnegative input vectors to image vectors, we formulate an NN-BP problem and provide a characterization of its optimal solution; see Lemma~\ref{lem 1}.	
	\item \emph{Equivalent Capacity Expression:} Leveraging the NN-BP characterization, we establish an equivalent expression for the capacity of MIMO-OWC channels under an average-power constraint;
	see Proposition~\ref{prop1}.
	\item \emph{Capacity Bounds:} Based on the equivalent capacity representation, we obtain a capacity lower bound by constructing a compound exponential distribution and develop capacity upper bounds via the duality approach; see Theorems~\ref{thm: lb}, \ref{thm: ub}, and \ref{thm: case II}.
	\item \emph{High-SNR Asymptotics:} By comparing the proposed lower and upper bounds, we establish the high-SNR asymptotic capacities in Theorems~\ref{thm: hsnr} and \ref{thm: case II}, thereby closing the constant gap left in the existing literature.
\end{itemize}




Sec.~\ref{sec: channel model} introduces the channel model.
Sec.~\ref{sec: NN-BP} presents the NN-BP results. Secs.~\ref{sec: nt>nr} and \ref{sec: nt<nr} characterize the channel capacities for $n_\textnormal{T}\geq n_\textnormal{R}$ and $n_\textnormal{T}< n_\textnormal{R}$, respectively. Numerical results are given in Sec.~\ref{sec: numerical results} and conclusions are drawn in Sec.~\ref{sec:conclusion}. Most proofs are deferred to the appendices.

\emph{Notation:} We use regular letters for scalars, e.g., $X$ denotes a random scalar and $x$ its realization. Boldfaced letters are used for vectors, e.g., $\mathbf{X}$ denotes a random vector and $\mathbf{x}$ its realization. Blackboard-bold letters are used for matrices, e.g., $\mathbb{H}$. Mutual information is denoted by $\const{I}(\cdot;\cdot)$, entropy by $\const{H}(\cdot)$, differential entropy by $\const{h}(\cdot)$, Kullback-Leibler divergence by $\textsf{D}(\cdot\|\cdot)$, and expectation by $\textsf{E}\{\cdot\}$. The $\ell_1$-norm is denoted by $\lVert\cdot\rVert_1$, $Q$-function by $Q(\cdot)$,
and natural logarithm by $\log(\cdot)$. $\mathfrak{R}$ denotes the set of real numbers. The $n$-dimensional volume of a set is denoted by $\textnormal{Vol}_n(\cdot)$, the closure by $\operatorname{cl}(\cdot)$, and the interior by $\operatorname{int}(\cdot)$. $\mathbb{I}_{n}$, $\mathbf{0}_n$ and $\mathbf{1}_n$ denote $n\times n$ identity matrix, $n$-dimensional all-zero and all-one vectors, respectively.
%

\section{Channel Model}\label{sec: channel model}
Consider an OWC channel with $n_\textnormal{T}$ transmit apertures and $n_\textnormal{R}$ receive apertures, the channel output is given by
\begin{IEEEeqnarray}{rCl}
	\mathbf{Y} = \mathbb{H} \mathbf{X} + \mathbf{Z}, \label{eq: channel model}
\end{IEEEeqnarray}
where the channel matrix $\mathbb{H}\in \mathfrak{R}^{n_\textnormal{R}\times n_\textnormal{T}}_+$, the channel input $\mathbf{X}\in\mathfrak{R}^{n_\textnormal{T}}$, the channel output $\mathbf{Y}\in\mathfrak{R}^{n_\textnormal{R}}$, and the channel noise $\mathbf{Z}\in\mathfrak{R}^{n_\textnormal{R}}$. The additive noise $\mathbf{Z}$ originates predominantly from shot noise and thermal noise, and approximately follows Gaussian distribution, i.e., $\mathbf{Z}\sim\mathcal{N}(\mathbf{0}_{n_\textnormal{R}},\sigma^2 \mathbb{I}_{n_\textnormal{R}})$.

Since the channel input represents optical power, it is constrained to be nonnegative \cite{Zhou2017}, i.e.,
\begin{IEEEeqnarray}{rCl}
	\mathrm{Pr}\{\mathbf{X} \in  \mathfrak{R}_+^{n_\textnormal{T}}\}=1. \label{eq: non-negative}
\end{IEEEeqnarray}
Moreover, due to safety and power-consumption requirements, the average power of the channel input is also constrained \cite{Chaaban2017}. We consider a total average-power constraint:
\begin{IEEEeqnarray}{rCl}
	\textsf{E} \{ \| \mathbf{X} \|_1 \}\leq \EE. \label{eq: ave cons}
\end{IEEEeqnarray}
By \cite{Cover2006}, the channel capacity is given by 
\begin{IEEEeqnarray}{rCl}
	\const{C} = \max_{p_\mathbf{X} \textnormal{satisfies } \eqref{eq: non-negative} \textnormal{ and } \eqref{eq: ave cons}} \const{I}(\mathbf{X};\mathbf{Y}). \label{eq: capacity express}	
\end{IEEEeqnarray}

In this paper, $\mathbb{H}$ is assumed to be deterministic and full-rank \cite{Moser2017-ISIT,Chaaban2018-hsnr}. The proposed results apply to two cases: {Case I}: $\ n_\textnormal{T}\geq n_\textnormal{R}$; {Case II}: $n_\textnormal{T}< n_\textnormal{R}$.
Since the method proposed in the next section is particularly suitable for $n_\textnormal{T} \geq n_\textnormal{R}$, the following analysis first concentrates on this case. We extend the results to the case $n_\textnormal{T} < n_\textnormal{R}$ in Sec.~\ref{sec: nt<nr}. For convenience, we denote the channel matrix $\mathbb{H}$ by
\begin{IEEEeqnarray}{rCl}
	\mathbb{H}=\left(\mathbf{h}_1,\cdots,\mathbf{h}_{n_\textnormal{T}}\right).
\end{IEEEeqnarray}
We further define the following operator for any $r\times s$ matrix $\mathbb{M}=\left(\mathbf{m}_1,\cdots,\mathbf{m}_s\right)$ with full row rank:
\begin{align}
	&\mathcal{S}(\mathbb{M}) = \left\{ \sum_{i=1}^{s} a_i \mathbf{m}_i:
	a_1,\cdots,a_s\in\mathfrak{R}_+\right\}. \label{eq: S(H)}
\end{align}
We define the set of all choices of $r$ columns of $\mathbb{M}$ that are linearly independent by
\begin{IEEEeqnarray}{rCl}
	\mathscr{R}(\mathbb{M})
	&=&\Bigl\{\mathcal{I}=\{ i_1,i_2,\cdots, i_r \}: \nonumber\\
	&&\,\,\{ i_1,i_2,\cdots, i_r \}\subseteq \{1,2,\cdots,s\},\textnormal{ and}\nonumber\\ 	
	&& \,\,\mathbf{m}_{i_1},\mathbf{m}_{i_2},\cdots,\mathbf{m}_{ i_r } \textnormal{ are linearly independent} \Bigr\}.\quad \label{eq: mathscr{R}(H)}
\end{IEEEeqnarray}
For $\mathcal{I}=\{i_1,\cdots,i_r\}\in\mathscr{R}(\mathbb{M})$, we define the submatrix:
\begin{align}
	\mathbb{M}_\mathcal{I} = \left(\mathbf{m}_{i_1},\mathbf{m}_{i_2},\cdots,\mathbf{m}_{ i_r }\right).
\end{align}

\section{Nonnegative Basis Pursuit}\label{sec: NN-BP}
This section proposes the NN-BP scheme for the case $n_\textnormal{T}\geq n_\textnormal{R}$. We first rewrite the channel model in \eqref{eq: channel model} as
\begin{IEEEeqnarray}{rCl}
	\mathbf{Y} = \bar{\mathbf{X}} + \mathbf{Z}, \label{eq: channel model new}
\end{IEEEeqnarray}
with the image vector $\bar{\mathbf{X}}$ defined by
\begin{IEEEeqnarray}{rCl}
	\bar{\mathbf{X}} = \mathbb{H} \mathbf{X}. \label{eq: bar{x} and x}
\end{IEEEeqnarray}
Applying \eqref{eq: S(H)} to the channel matrix $\mathbb{H}$, we obtain
\begin{align}
	\mathcal{S}(\mathbb{H}) = \left\{ \sum_{i=1}^{n_\textnormal{T}} a_i \mathbf{h}_i:
	a_1,\cdots,a_{n_\textnormal{T}}\in\mathfrak{R}_+\right\}.
\end{align}
Note that $\mathcal{S}(\mathbb{H})$ is a convex cone. Since $\mathbf{x}\in\mathfrak{R}^{n_\textnormal{T}}_+$, the image vector $\bar{\mathbf{X}}$ must fall within $\mathcal{S}(\mathbb{H})$. For a realization $\bar{\mathbf{x}}$ of $\bar{\mathbf{X}}$, the NN-BP problem is given by
\begin{align}
	\bm{P}:\,\,\mathbf{x}^\star = 
	\argmin_{ \mathbf{x}\in \mathfrak{R}^{n_\textnormal{T}}_+ } \,\,\, &\|\mathbf{x} \|_1 \\\label{eq: minimum-energy problem} 
	\textnormal{s.t.} \,\,\, &\mathbb{H} \mathbf{x} = \bar{\mathbf{x}},
\end{align}
where $\mathbf{x}^\star=(x_1^\star,\cdots,x_{n_\textnormal{T}}^\star)^\textsf{T}$ is the optimal input that yields
$\bar{\mathbf{x}}$ with minimum $\ell_1$-norm. Compared with the classical BP problem \cite{Chen1998,Hyder2016}, the NN-BP problem promotes a sparsest \emph{nonnegative} input by minimizing the $\ell_1$-norm under an exact reconstruction constraint. To solve the NN-BP problem, we first illustrate the structure of $\mathbf{x}^\star$ through examples and then formalize the general solution.
 
\begin{figure}[t!]
	\centering
	\includegraphics[width=2.05in]{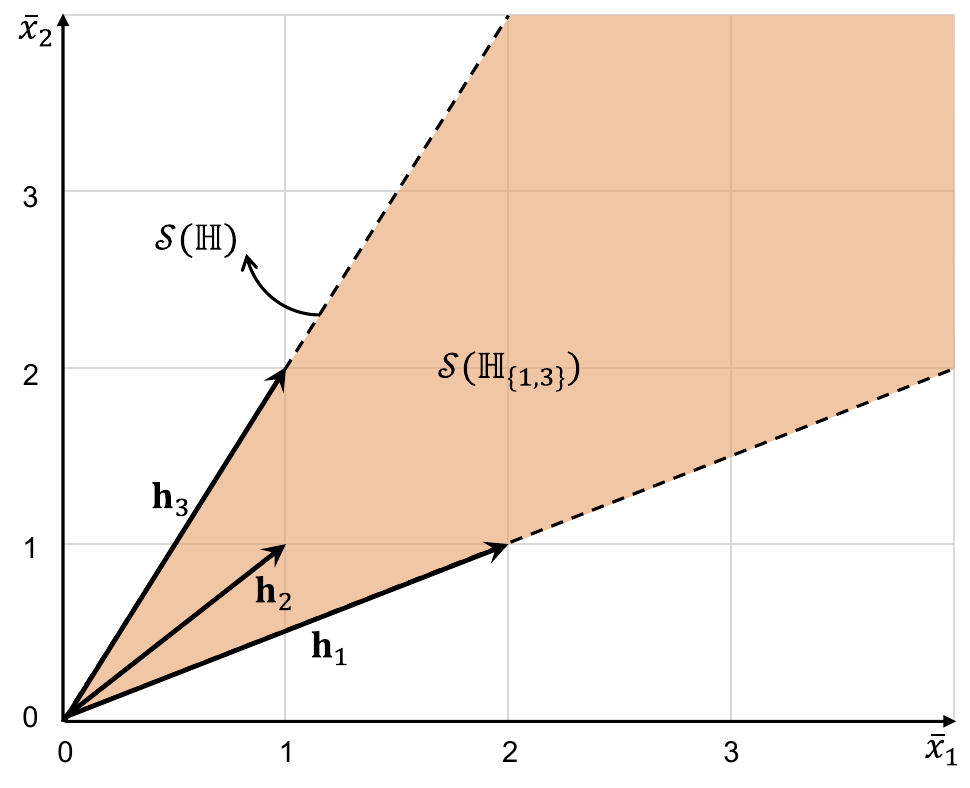}
	\caption{The cone $\mathcal{S}(\mathbb{H})$ when $\mathbb{H}=(2, 1, 1;1, 1, 2)$ and its partition into subcone $\mathcal{S}(\mathbb{H}_{\{1,3\}})$.}
	\label{fig: example1}
\end{figure}
\subsection{Examples of Typical Solutions}
\begin{exam}\label{example1}
Consider a $2\times3$ channel matrix
\begin{IEEEeqnarray}{rCl}
	\mathbb{H}=\left( \begin{array}{ccc}
		2 & 1& 1\\
		1 & 1& 2
	\end{array}\right), \label{eq: example 1}
\end{IEEEeqnarray}
where $\mathbf{h}_1 = (2, 1)^{\textsf{T}}$, $\mathbf{h}_2 = (1,1)^{\textsf{T}}$, and $\mathbf{h}_3 = (1,2)^{\textsf{T}}$. Fig.~\ref{fig: example1} illustrates the region covered by $\mathcal{S}(\mathbb{H})$ in the two-dimensional plane. As shown in Fig.~\ref{fig: example1}, the whole cone $\mathcal{S}(\mathbb{H})$ can also be characterized by $\mathcal{S}(\mathbb{H}_{\{1,3\}})$, i.e.,
\begin{IEEEeqnarray}{rCl}
	\mathcal{S}(\mathbb{H})=\mathcal{S}(\mathbb{H}_{\{1,3\}}). \label{eq: partition ex2}
\end{IEEEeqnarray}
	
To compute the NN-BP solution $\mathbf{x}^\star$ in \eqref{eq: minimum-energy problem}, we first establish the following relationships among $\mathbf{h}_1$, $\mathbf{h}_2$, and $\mathbf{h}_3$:
	\begin{subequations}
		\label{eq: h1h2h3 relation}
		\begin{IEEEeqnarray}{rCc C c}
			\mathbf{h}_1 &=& 3 \mathbf{h}_2 &-&  \mathbf{h}_3, \label{eq: h1toh2h3}\\
			\mathbf{h}_2 &=& \frac{1}{3} \mathbf{h}_1 &+&  \frac{1}{3} \mathbf{h}_3, \label{eq: h2toh1h3}\\
			\mathbf{h}_3 &=& - \mathbf{h}_1 &+&  3 \mathbf{h}_2.\label{eq: h3toh1h2}
		\end{IEEEeqnarray}
	\end{subequations}
For any $\bar{\mathbf{x}}\in\mathcal{S}(\mathbb{H})$, there exists $\mathbf{x}^{\{1,2,3\}}=(x_1,x_2,x_3)^{\textsf{T}}\in\mathfrak{R}^3_+$ such that
	\begin{IEEEeqnarray}{rCl}
		\bar{\mathbf{x}} = x_1 \mathbf{h}_1 + x_2 \mathbf{h}_2 + x_3 \mathbf{h}_3. \label{eq: x1x2x3}
	\end{IEEEeqnarray}
Using \eqref{eq: h1h2h3 relation}, $\bar{\mathbf{x}}$ can also be represented with respect to the bases $\{\mathbf{h}_1,\mathbf{h}_2\}$, $\{\mathbf{h}_2,\mathbf{h}_3\}$, and $\{\mathbf{h}_1,\mathbf{h}_3\}$ as follows:
	\begin{IEEEeqnarray}{r C c C c}
		\subnumberinglabel{eq: bar{x} rewritten}	
		\bar{\mathbf{x}} 
		&=& (x_1- x_3) \mathbf{h}_1 &+& (x_2 + 3 x_3) \mathbf{h}_2,	 \label{eq: x1x2}\\
		\bar{\mathbf{x}}
		&=& (x_2+3x_1) \mathbf{h}_2 &+& (x_3-x_1)  \mathbf{h}_3,  \label{eq: x1x3} \\
		\bar{\mathbf{x}}
		&=& (x_1+\frac{1}{3}x_2) \mathbf{h}_1 &+& (x_3+\frac{1}{3}x_2) \mathbf{h}_3.  \label{eq: x2x3}
	\end{IEEEeqnarray}
Therefore, $\mathbf{x}^\star$ can be either $\mathbf{x}^{\{1,2,3\}}$ or one of the following possible forms:
	\begin{IEEEeqnarray}{rCl}
		\subnumberinglabel{eq: x_min rewritten}
		\mathbf{x}^{\{1,2\}}
		&=& \bigl(x_1-x_3,x_2+ 3x_3,0 \bigr)^\textsf{T}, \label{eq: x_min^{1,2}}\\
		\mathbf{x}^{\{2,3\}}
		&=& \bigl( 0,x_2+3x_1,x_3-x_1 \bigr)^\textsf{T}, \label{eq: x_min^{2,3}}\\
		\mathbf{x}^{\{1,3\}}
		&=& \bigl( x_1+\frac{1}{3}x_2,0,x_3+\frac{1}{3}x_2 \bigr)^\textsf{T}. \label{eq: x_min^{1,3}}
	\end{IEEEeqnarray}
Since	
	\begin{align}
		\|\mathbf{x}^{\{1,3\}} \|_1 \leq \| \mathbf{x}^{\{1,2,3\}} \|_1 \leq  \min\{\| \mathbf{x}^{\{1,2\}} \|_1, \| \mathbf{x}^{\{2,3\}} \|_1\}, \label{eq: comparison1} 
	\end{align}
we obtain that for any $\bar{\mathbf{x}}\in\mathcal{S}(\mathbb{H})$, $\mathbf{x}^{\{1,3\}}$ has minimum $\ell_1$-norm among all feasible inputs. Hence, the NN-BP solution $\mathbf{x}^\star$ is given by $\mathbf{x}^{\{1,3\}}$.	
\hfill $\blacklozenge$
\end{exam}

\begin{figure}[t]
	\centering
	\includegraphics[width=2.05in]{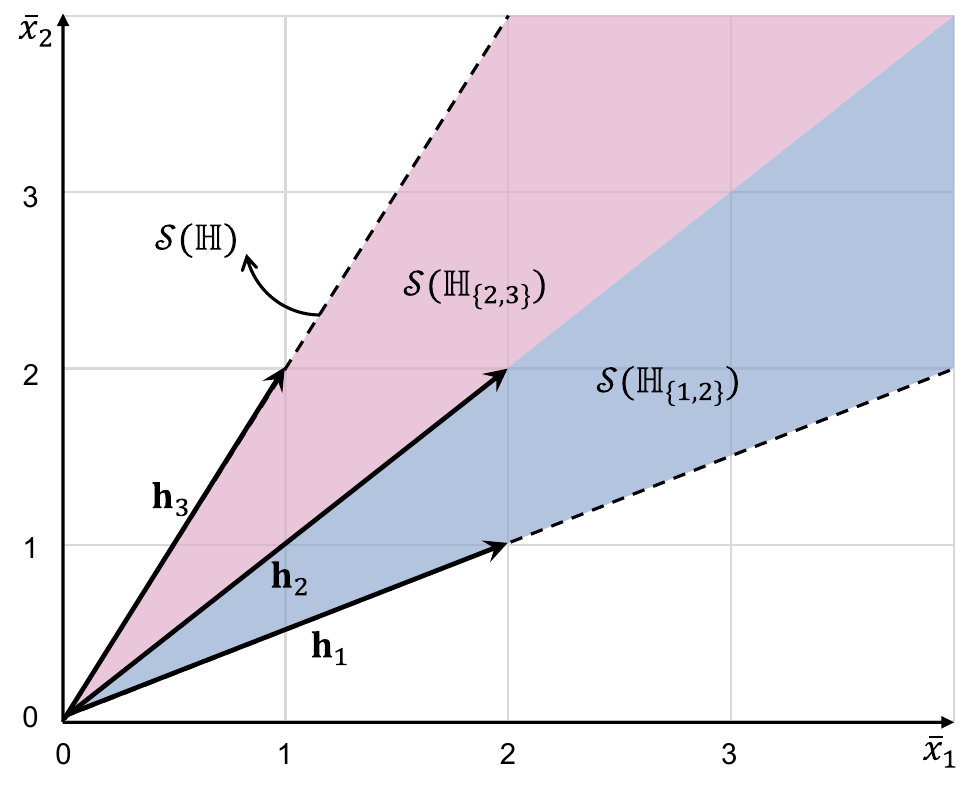}
	\caption{The cone $\mathcal{S}(\mathbb{H})$ when $\mathbb{H} = (2, 2, 1; 1, 2, 2)$ and its partition into subcone $\mathcal{S}(\mathbb{H}_{\{1,2\}})$ and $\mathcal{S}(\mathbb{H}_{\{2,3\}})$.}
	\label{fig: example2}
	\vspace{-1.2mm}
\end{figure}

\begin{exam}\label{example2} 
Consider another $2\times3$ channel matrix
\begin{IEEEeqnarray}{rCl}
	\mathbb{H}=\left( \begin{array}{ccc}
		2 & 2 & 1\\
		1 & 2 & 2
	\end{array}\right), \label{eq: example2}
\end{IEEEeqnarray}
where $\mathbf{h}_1$ and $\mathbf{h}_3$ are preserved from Example~\ref{example1}, but the length of $\mathbf{h}_2$ is deliberately increased. Fig.~\ref{fig: example2} illustrates the region covered by $\mathcal{S}(\mathbb{H})$, which is constituted by the the union of the pink and blue shaded regions. The relationships among $\mathbf{h}_1$, $\mathbf{h}_2$, and $\mathbf{h}_3$ are given by
\begin{IEEEeqnarray}{rCrCl}
	\subnumberinglabel{eq: h1h2h3 relation ex2}
	\mathbf{h}_1 &=& \dfrac{3}{2}\,\mathbf{h}_2 &-& \mathbf{h}_3, \label{eq: h1toh2h3 new}\\
	\mathbf{h}_2 &=& \dfrac{2}{3}\,\mathbf{h}_1 &+& \frac{2}{3}\,\mathbf{h}_3, \label{eq: h2toh1h3 new}\\
	\mathbf{h}_3 &=& -\mathbf{h}_1 &+& \frac{3}{2}\,\mathbf{h}_2. \label{eq: h3toh1h2 new}
\end{IEEEeqnarray}
Following a similar argument as in Example~\ref{example1}, we can show that $\mathbf{x}^\star$ can be either $\mathbf{x}^{\{1,2,3\}}$ or one of the following forms:
\begin{IEEEeqnarray}{rCl}	
	\subnumberinglabel{eq: x_min rewritten ex2}
	\mathbf{x}^{\{1,2\}}
	&=& \bigl(x_1-x_3,x_2+ \frac{3}{2}x_3,0 \bigr)^\textsf{T},\\
	\mathbf{x}^{\{2,3\}}
	&=& \bigl( 0,x_2+\frac{3}{2}x_1,x_3-x_1 \bigr)^\textsf{T},\\
	\mathbf{x}^{\{1,3\}}
	&=& \bigl( x_1+\frac{2}{3}x_2,0,x_3+\frac{2}{3}x_2 \bigr)^\textsf{T}.
\end{IEEEeqnarray}
We next identify the minimum-$\ell_1$-norm candidate in \eqref{eq: x_min rewritten ex2}, thereby determining $\mathbf{x}^\star$.
From the geometric relationships illustrated in Fig.~\ref{fig: example2}, if $\bar{\mathbf{x}}\in \mathcal{S}(\mathbb{H}_{\{1,2\}})$, then it can be shown that $\mathbf{x}^{\{1,2,3\}}$, $\mathbf{x}^{\{1,2\}}$, $\mathbf{x}^{\{1,3\}}\in\mathfrak{R}^3_+$, whereas $\mathbf{x}^{\{2,3\}}\notin\mathfrak{R}^3_+ $ since $x_3-x_1<0$. The $\ell_1$-norms of these feasible points satisfy 
\begin{align}
	\| \mathbf{x}^{\{1,2\}} \|_1  \leq \| \mathbf{x}^{\{1,2,3\}} \|_1 \leq 	\| \mathbf{x}^{\{1,3\}} \|_1. \label{eq: 29}
\end{align}
Consequently, for any $\bar{\mathbf{x}}\in\mathcal{S}(\mathbb{H}_{\{1,2\}})$, $\mathbf{x}^\star$ equals $\mathbf{x}^{\{1,2\}}$.

Similarly, if $\bar{\mathbf{x}}\in \mathcal{S}(\mathbb{H}_{\{2,3\}})$, then it can be shown that $\mathbf{x}^{\{1,2,3\}}$, $\mathbf{x}^{\{2,3\}}$, $\mathbf{x}^{\{1,3\}}\in\mathfrak{R}^3_+$, whereas $\mathbf{x}^{\{1,2\}}\notin\mathfrak{R}^3_+$ since $x_1-x_3<0$. The $\ell_1$-norms of these feasible points satisfy 
\begin{align}
	\| \mathbf{x}^{\{2,3\}} \|_1  \leq \| \mathbf{x}^{\{1,2,3\}} \|_1 \leq 	\| \mathbf{x}^{\{1,3\}} \|_1. \label{eq: 30}
\end{align}
Thus, for any $\bar{\mathbf{x}}\in\mathcal{S}(\mathbb{H}_{\{2,3\}})$, $\mathbf{x}^\star$ equals $\mathbf{x}^{\{2,3\}}$.

As illustrated in Fig.~\ref{fig: example2}, the whole cone $\mathcal{S}(\mathbb{H})$ can be partitioned into two subcones: $\mathcal{S}(\mathbb{H}_{\{1,2\}})$ and $\mathcal{S}(\mathbb{H}_{\{2,3\}})$, i.e.,
\begin{IEEEeqnarray}{rCl}
	\subnumberinglabel{eq: partition scheme ex1}
	&&\mathcal{S}(\mathbb{H})=\mathcal{S}(\mathbb{H}_{\{1,2\}})\cup\mathcal{S}(\mathbb{H}_{\{2,3\}}), \label{eq: partition ex1}\\
	&&\textnormal{Vol}_2\{\mathcal{S}(\mathbb{H}_{\{1,2\}})\cap\mathcal{S}(\mathbb{H}_{\{2,3\}})\}=0. \label{eq: partition two parts}
\end{IEEEeqnarray}
The above results indicate that, within each region $\mathcal{S}(\mathbb{H}_{\{1,2\}})$ and $\mathcal{S}(\mathbb{H}_{\{2,3\}})$, the NN-BP solution $\mathbf{x}^\star$ admits a region-specific but fixed form, given by $\mathbf{x}^{\{1,2\}}$ and $\mathbf{x}^{\{2,3\}}$, respectively.
\hfill $\blacklozenge$
\end{exam}


\begin{exam}\label{example3}
Consider another $2\times3$ channel matrix
\begin{IEEEeqnarray}{rCl}
	\mathbb{H}=\left( \begin{array}{ccc}
		2 & 1.5 & 1\\
		1 & 1.5 & 2
	\end{array}\right), \label{eq: example 3}
\end{IEEEeqnarray}
where the endpoints of $\mathbf{h}_1$, $\mathbf{h}_2$, and $\mathbf{h}_3$ are collinear, as illustrated in Fig.~\ref{fig3}. 
Their relationships can be expressed as
\begin{IEEEeqnarray}{rCcCc}
	\subnumberinglabel{eq: h1h2h3 relation ex3}
	\mathbf{h}_1 &=& 2 \mathbf{h}_2 &-&  \mathbf{h}_3,\\
	\mathbf{h}_2 &=& \frac{1}{2} \mathbf{h}_1 &+& \frac{1}{2} \mathbf{h}_3,\\
	\mathbf{h}_3 &=& - \mathbf{h}_1 &+& 2 \mathbf{h}_2.
\end{IEEEeqnarray}
It follows that $\mathbf{x}^\star$ can be either $\mathbf{x}^{\{1,2,3\}}$ or one of the following forms:
\begin{IEEEeqnarray}{rCl}
	\subnumberinglabel{eq: x_min rewritten ex3}
	\mathbf{x}^{\{1,2\}}
	&=& \bigl(x_1-x_3,x_2+ 2x_3,0 \bigr)^\textsf{T},\\
	\mathbf{x}^{\{2,3\}}
	&=& \bigl( 0,x_2+2x_1,x_3-x_1 \bigr)^\textsf{T},\\
	\mathbf{x}^{\{1,3\}}
	&=& \bigl( x_1+\frac{1}{2}x_2,0,x_3+\frac{1}{2}x_2 \bigr)^\textsf{T}.
\end{IEEEeqnarray}
If $\bar{\mathbf{x}}\in\mathcal{S}(\mathbb{H}_{\{1,2\}})$, then the feasible points in $\mathfrak{R}^3_+$ satisfy
\begin{align}
	\| \mathbf{x}^{\{1,2\}} \|_1= \| \mathbf{x}^{\{1,3\}} \|_1 = \| \mathbf{x}^{\{1,2,3\}} \|_1. \label{eq: 38}
\end{align}
If $\bar{\mathbf{x}}\in\mathcal{S}(\mathbb{H}_{\{2,3\}})$, then the feasible points in $\mathfrak{R}^3_+$ satisfy
\begin{align}
	\| \mathbf{x}^{\{2,3\}} \|_1= \| \mathbf{x}^{\{1,3\}} \|_1 = \| \mathbf{x}^{\{1,2,3\}} \|_1. \label{eq: 39}
\end{align}

\begin{figure}[t]
	\centering	
	\subfloat[Partition into the subcone $\mathcal{S}(\mathbb{H}_{\{1,3\}})$.]{
		\includegraphics[width=1.58in]{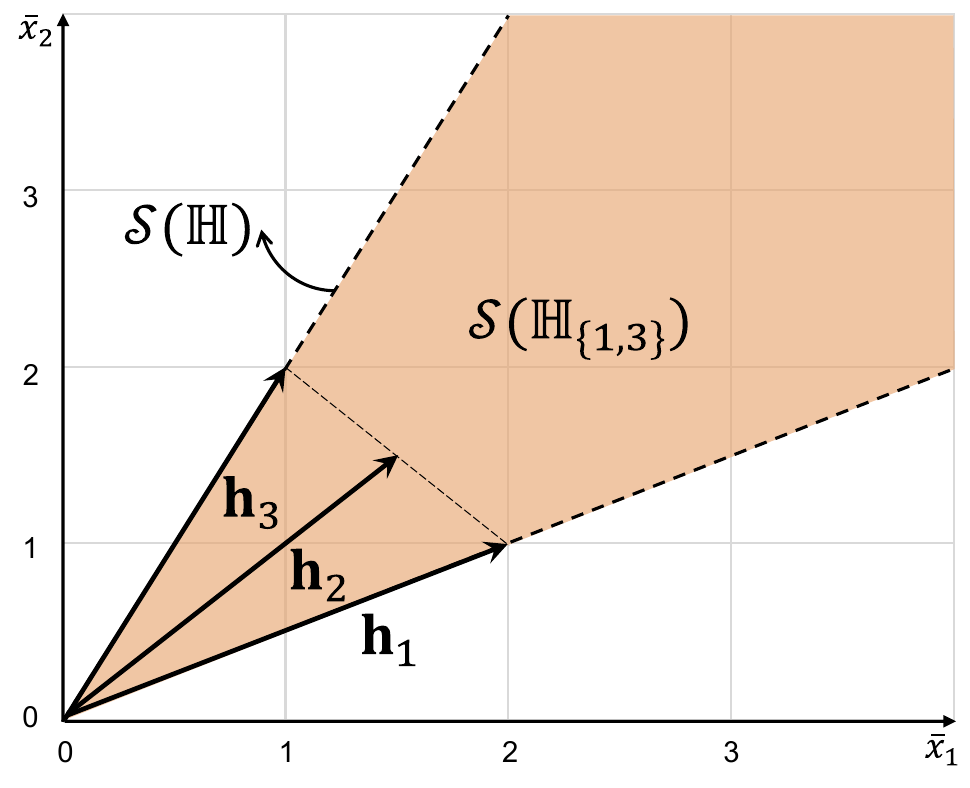}\label{fig3a}
	}\hspace{+2mm}	
	\subfloat[Partition into two subcones $\mathcal{S}(\mathbb{H}_{\{1,2\}})$ and $\mathcal{S}(\mathbb{H}_{\{2,3\}})$.]{
		\includegraphics[width=1.58in]{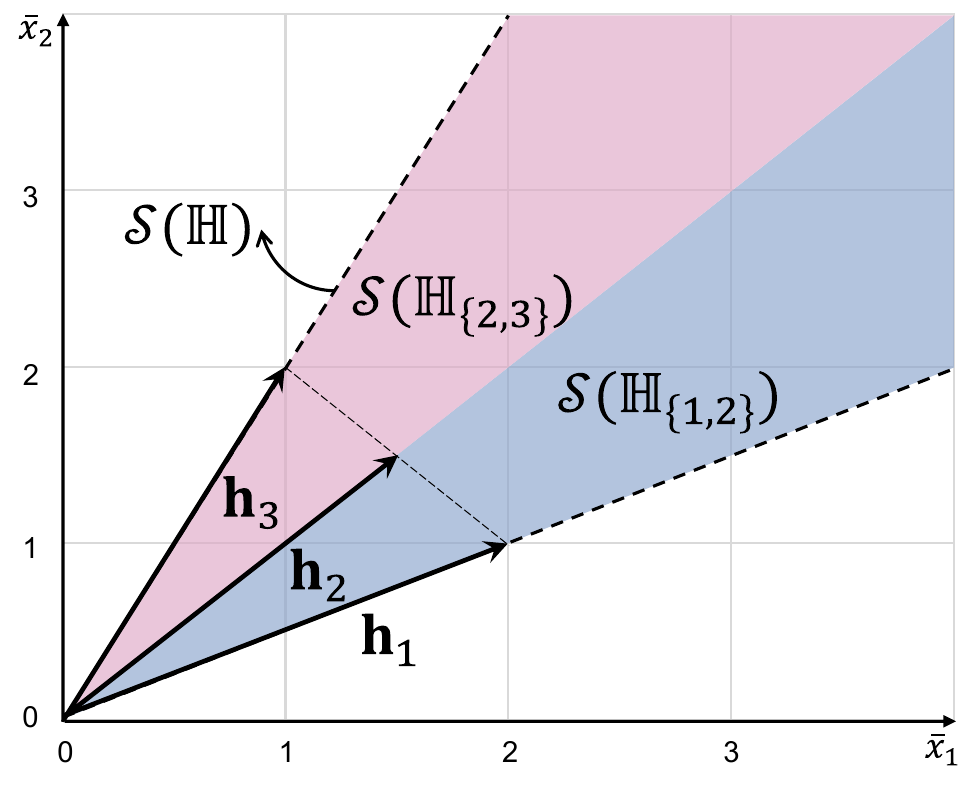}\label{fig3b}
	}
	\caption{The cone $\mathcal{S}(\mathbb{H})$ when $\mathbb{H} = (2, 1.5, 1; 1, 1.5, 2)$ and its partitions.}
	\label{fig3}
\end{figure}

On the one hand, Fig.~\ref{fig3a} illustrates the partition of $\mathcal{S}(\mathbb{H})$ by the subcone $\mathcal{S}(\mathbb{H}_{\{1,3\}})$. \eqref{eq: 38} and \eqref{eq: 39} show that the NN-BP solution $\mathbf{x}^\star$ can be represented by $\mathbf{x}^{\{1,3\}}$. On the other hand, Fig.~\ref{fig3b} presents an alternative partition of $\mathcal{S}(\mathbb{H})$ by two subcones $\mathcal{S}(\mathbb{H}_{\{1,2\}})$ and $\mathcal{S}(\mathbb{H}_{\{2,3\}})$. \eqref{eq: 38} and \eqref{eq: 39} show that, within each region $\mathcal{S}(\mathbb{H}_{\{1,2\}})$ and $\mathcal{S}(\mathbb{H}_{\{2,3\}})$, the corresponding $\mathbf{x}^\star$ can be $\mathbf{x}^{\{1,2\}} $ and $\mathbf{x}^{\{2,3\}} $, respectively. This example demonstrates that both partition schemes in \eqref{eq: partition ex2} and \eqref{eq: partition scheme ex1} are admissible. Consequently, the representation of $\mathbf{x}^\star$ is not necessarily unique for this channel.
\hfill $\blacklozenge$
\end{exam}

\vspace{-2mm}
\subsection{Characterization of General NN-BP Solutions}
Motivated by Examples~\ref{example1}, \ref{example2}, and \ref{example3}, we propose the following lemma for solving the NN-BP problem with an arbitrary full-row-rank channel matrix $\mathbb{H}$. The proof is deferred to Appendix~\ref{app: proof of lem 1}. Before stating the lemma, we introduce some necessary notation. For any $\mathcal{I}\in\mathscr{R}(\mathbb{H})$, define the following $n_\textnormal{T}-n_\textnormal{R}$ vectors: 
\begin{align}
	\mathbf{a}_{\mathcal{I},j} = \mathbb{H}_\mathcal{I}^{-1} \mathbf{h}_j, \quad j\in\mathcal{I}^c,
\end{align}
and $n_\textnormal{T}-n_\textnormal{R}$ scalars:
\begin{align}
	\gamma_{\mathcal{I},j} = \mathbf{1}_{n_\textnormal{R}}^\textsf{T} \mathbf{a}_{\mathcal{I},j},\quad j\in\mathcal{I}^c, \label{eq: gamma def}
\end{align}
where $\mathcal{I}^c$ is the complement of $\mathcal{I}$, i.e., $\mathcal{I}^c=\{1,\cdots,n_\textnormal{T}\} \setminus \mathcal{I}$.

\begin{lem}[NN-BP Result]\label{lem 1}

\begin{itemize}

\item \emph{Part I:}

The cone $\mathcal{S}(\mathbb{H})$ can be partitioned into a collection of subcones $\{\mathcal{S}(\mathbb{H}_\mathcal{I}):\mathcal{I}\in\mathscr{S}(\mathbb{H})\}$, i.e.,
\begin{align}
	& \mathcal{S}(\mathbb{H}) = \cup_{\mathcal{I}\in\mathscr{S}(\mathbb{H})} \mathcal{S}(\mathbb{H}_\mathcal{I}), \label{eq: lem5-1}\\
	& \textnormal{Vol}_{n_\textnormal{R}} \left\{ \mathcal{S}(\mathbb{H}_\mathcal{I}) \cap \mathcal{S}(\mathbb{H}_\mathcal{J})\right\} =0, \nonumber\\
	&\qquad\qquad\qquad\qquad\, \forall\mathcal{I},\mathcal{J}\in\mathscr{S}(\mathbb{H}),\, \mathcal{I}\neq\mathcal{J}, \label{eq: lem5-2}
\end{align}
where the set $\mathscr{S}(\mathbb{H})$ is constructed from $\mathscr{R}(\mathbb{H})$ and characterized as follows:
\begin{itemize}
	\item For any $\mathcal{I} \in \mathscr{R}(\mathbb{H})$ and $ j \in \mathcal{I}^c$, if	
	\begin{align}
		\gamma_{\mathcal{I},j} \neq 1, \label{eq: gamma_{I,j} neq 1}
	\end{align}
	then 
	\begin{align}
		\mathscr{S}(\mathbb{H}) = \{\mathcal{I}\in\mathscr{R}(\mathbb{H}): \gamma_{\mathcal{I},j} < 1, \, \forall j \in \mathcal{I}^c \}. \label{eq: mathscr{S}(H)}
	\end{align}
	\item If \eqref{eq: gamma_{I,j} neq 1} is violated, then $\mathscr{S}(\mathbb{H}) $ has multiple solutions. Algorithm~\ref{alg 1} simply picks one out of all possible solutions.
\end{itemize}

\item \emph{Part II:}\\
For any $\mathcal{I}\in\mathscr{S}(\mathbb{H}) $, if $\bar{\mathbf{x}}\in\mathcal{S}(\mathbb{H}_\mathcal{I})$, then the optimal solution $\mathbf{x}^\star$ to the NN-BP problem is given by
\begin{IEEEeqnarray}{rCl}
	\subnumberinglabel{eq: opt x express}
	&&\mathbf{x}_\mathcal{I}^\star = (\mathbb{H}_{\mathcal{I}})^{-1} \bar{\mathbf{x}},\\
	&&x_j^\star = 0, \quad\forall j\in \mathcal{I}^c,
\end{IEEEeqnarray}
where 
\begin{IEEEeqnarray}{ll}
	\mathbf{x}_\mathcal{I}^\star=(x_{i_1}^\star,\cdots,x_{i_{n_\textnormal{R}}}^\star)^\textsf{T}, &
	\quad\mathcal{I}=\{i_1,\cdots,i_{n_\textnormal{R}}\}.
\end{IEEEeqnarray}

\end{itemize}

\end{lem}


\begin{algorithm}
	\caption{}
	\label{alg 1}
	\renewcommand{\algorithmicrequire}{\textbf{Input:}}
	\renewcommand{\algorithmicensure}{\textbf{Output:}}
\begin{algorithmic}[1]
	\REQUIRE $\mathbb{H}$
	\ENSURE $\mathscr{S}(\mathbb{H})$
	\STATE Initialize $\mathscr{S}(\mathbb{H})=\mathscr{R}(\mathbb{H})$ by \eqref{eq: mathscr{R}(H)};
	\FOR{$j \in \{1,\cdots,n_\textnormal{T}\}$}
	\FOR{ $\mathcal{I} \in \mathscr{R}(\mathbb{H})$ and $\mathcal{I} \subseteq \{j,\cdots,n_\textnormal{T}\}$}
	\IF{$j\in\mathcal{I}^c$ and $\gamma_{\mathcal{I},j}\geq1$}
	\STATE Update $\mathscr{S}(\mathbb{H})$:  $\mathscr{S}(\mathbb{H})=\mathscr{S}(\mathbb{H})\setminus \mathcal{I}$;
	\ELSE
	\FOR{$k\in\mathcal{I}^c \cap \{j+1,\cdots,n_\textnormal{T}\}$}
	\IF{$\gamma_{\mathcal{I},k}>1$ or \{$\gamma_{\mathcal{I},k}=1$ and the first component of $\mathbf{a}_{\mathcal{I},k}$ is negative\}}
	\STATE Update $\mathscr{S}(\mathbb{H})$:  $\mathscr{S}(\mathbb{H})=\mathscr{S}(\mathbb{H})\setminus \mathcal{I}$;
	\ENDIF
	\ENDFOR
	\ENDIF
	\ENDFOR
	\ENDFOR
\end{algorithmic}
\end{algorithm}

\begin{rmk}
The NN-BP result in the above lemma plays a role analogous to the minimum-energy signaling approach in \cite{Li2020}, but it leads to a different geometric partition. In \cite{Li2020}, the peak-power constraint confines the input space to $[0,\amp]^{n_\textnormal{T}}$. After selecting an active subset of $n_\textnormal{R}$ linearly independent channel columns, the remaining $n_\textnormal{T}-n_\textnormal{R}$ input components are assigned to the boundary values $0$ or $\amp$. As a result, the bounded image zonotope is covered by $|\mathscr{R}(\mathbb{H})|$ shifted bounded parallelepipeds, where $|\mathscr{R}(\mathbb{H})|$ denotes the number of linearly independent $n_\textnormal{R}$-column groups. This construction crucially depends on the finite amplitude boundary $\amp$.
	
However, in the setting with only an average-power constraint, no such finite boundary exists. The image space induced by nonnegative inputs is an unbounded cone rather than a bounded zonotope, and the result in \cite{Li2020} therefore does not extend to the present channel. Instead, the proposed NN-BP result divides the image cone into $|\mathscr{S}(\mathbb{H})|$ unbounded subcones, each associated with an active aperture subset selected by the minimum-$\ell_1$-norm representation. Compared with $|\mathscr{R}(\mathbb{H})|$, $|\mathscr{S}(\mathbb{H})|$ can be strictly smaller, as shown in Examples \ref{example1}--\ref{example3}, where $|\mathscr{R}(\mathbb{H})|=3$ whereas $|\mathscr{S}(\mathbb{H})|=1$ or $2$, depending on the NN-BP partition returned by Algorithm~\ref{alg 1}.

\hfill $\blacklozenge$
\end{rmk}


\section{Main Results For $n_\textnormal{T}\geq n_\textnormal{R}$}\label{sec: nt>nr}
This section applies the NN-BP result to characterize the channel capacity for $n_\textnormal{T}\geq n_\textnormal{R}$. We begin by defining a random variable $\widetilde{\mathcal{I}}$ over $\mathscr{S}(\mathbb{H})$, which specifies the subcone $\mathcal{S}(\mathbb{H}_{\widetilde{\mathcal{I}}})$ containing $\bar{\mathbf{X}}$, i.e.,
\begin{align}
	\bar{\mathbf{X}}\in\mathcal{S}(\mathbb{H}_\mathcal{I}) \Longleftrightarrow \widetilde{\mathcal{I}} = \mathcal{I}.
\end{align}
Hence, $\widetilde{\mathcal{I}}$ can be regarded as a deterministic function of $\bar{\mathbf{X}}$.\footnote{If $\bar{\mathbf{X}}$ lies on the common boundary of different subcones $\mathcal{S}(\mathbb{H}_\mathcal{I})$ and $\mathcal{S}(\mathbb{H}_\mathcal{J})$, then $\widetilde{\mathcal{I}}$ can be chosen arbitrarily from $\mathcal{I}$ and $\mathcal{J}$ without affecting the results.} 
Let $\mathbf{p}$ be a probability vector over $\mathscr{S}(\mathbb{H})$ with entries:
\begin{align}
	p_\mathcal{I} &= \operatorname{Pr}\{\widetilde{\mathcal{I}}=\mathcal{I}\}, \quad \mathcal{I}\in\mathscr{S}(\mathbb{H}).	\label{eq: p pdf}
\end{align}
Moreover, define a probability vector $\mathbf{q}$ over $\mathscr{S}(\mathbb{H})$ by
\begin{align}	
	q_{\mathcal{I}} &= \frac{\left|\operatorname{det}\left(\mathbb{H}_\mathcal{I}\right)\right|}{\sum\nolimits_{\mathcal{I} \in \mathscr{S}(\mathbb{H}) } \left|\operatorname{det}\left(\mathbb{H}_\mathcal{I}\right)\right|}, \quad \mathcal{I}\in\mathscr{S}(\mathbb{H}). \label{eq: q vec def}
\end{align}
Using $\widetilde{\mathcal{I}}$, we establish the following proposition, whose proof is provided in Appendix~\ref{app: proof of prop}.
\begin{prop}\label{prop1}
The channel capacity in \eqref{eq: capacity express} can be equivalently expressed as
\begin{IEEEeqnarray}{rCl}
	\const{C} = \max_{ p_{\bar{\mathbf{X}}} } \const{I}(\bar{\mathbf{X}};\mathbf{Y}), \label{eq: prop eq1}	
\end{IEEEeqnarray}
with $p_{\bar{ \mathbf{X} }} $ being the distribution over $\mathcal{S}( \mathbb{H})$ and satisfying:
\begin{align}
	\textsf{E}_{\widetilde{\mathcal{I}}}\Bigl\{ \bigl\|\mathbb{H}_{\widetilde{\mathcal{I}}}^{-1} \textsf{E}_{\bar{\mathbf{X}}} \bigl\{ \bar{\mathbf{X}}|\widetilde{\mathcal{I}} \bigr\} \bigr\|_1 \Bigl\} \leq \EE. \label{eq: prop eq2}
\end{align}
\end{prop}

We now proceed to derive capacity bounds based on Proposition~\ref{prop1}. The main results are summarized in the following theorems, with proofs provided in Appendix~\ref{app: proof of case i}.

First, we choose the distribution of $\widetilde{\mathcal{I}}$ according to \eqref{eq: q vec def}, i.e.,
\begin{align}
	p_{\mathcal{I}}=q_{\mathcal{I}}, \quad \mathcal{I}\in\mathscr{S}(\mathbb{H}). \label{eq: q=p}
\end{align}
For any $\mathcal{I}\in\mathscr{S}(\mathbb{H})$, let $\bar{\mathbf{X}}$ be exponentially distributed over $\mathcal{S}\left(\mathbb{H}_\mathcal{I}\right)$, i.e.,
\begin{align}
	&p_{\bar{\mathbf{X}}|\widetilde{\mathcal{I}}=\mathcal{I}} (\bar{\mathbf{x}})\nonumber\\
	&=\frac{n_\textnormal{R}^{n_\textnormal{R}}}{\left|\operatorname{det}\left(\mathbb{H}_\mathcal{I}\right)\right|\EE^{n_\textnormal{R}}} \exp \left(-\frac{n_\textnormal{R}}{\EE}\left\|\mathbb{H}_\mathcal{I}^{-1} \bar{\mathbf{x}}\right\|_1\right),\,\bar{\mathbf{x}} \in \mathcal{S}\left(\mathbb{H}_\mathcal{I}\right). \label{eq: exp pdf}
\end{align}
It can be shown that
{\setlength\abovedisplayskip{4.5pt} 
	\setlength\belowdisplayskip{4.5pt}
\begin{align}	
	\textsf{E}_{\bar{\mathbf{X}}}
	\left\{\bar{\mathbf{X}} \mid \widetilde{\mathcal{I}}=\mathcal{I}\right\}
		&=\int_{\bar{\mathbf{x}}\in\mathcal{S}(\mathbb{H}_\mathcal{I})} \bar{\mathbf{x}} \times p_{\bar{\mathbf{X}}|\widetilde{\mathcal{I}}=\mathcal{I}} (\bar{\mathbf{x}}) d \bar{x}_1\cdots d\bar{x}_{n_\textnormal{R}}\\
	&=\frac{\EE}{n_\textnormal{R}} \mathbb{H}_\mathcal{I} \mathbf{1}_{n_\textnormal{R}}.
\end{align}}
Then, the constraint \eqref{eq: prop eq2} is satisfied because
\begin{align}
	&\textsf{E}_{\widetilde{\mathcal{I}}}
	\left\{\left\|\mathbb{H}_{\widetilde{\mathcal{I}}}^{-1} \textsf{E}_{\bar{\mathbf{X}}}\{\bar{\mathbf{X}} \mid \widetilde{\mathcal{I}}\}\right\|_1\right\} \\
	& =\sum_{\mathcal{I} \in \mathscr{S}(\mathbb{H})} p_\mathcal{I} \times\left\|\mathbb{H}_\mathcal{I}^{-1} \textsf{E}_{\bar{\mathbf{X}}}\{\bar{\mathbf{X}} \mid \widetilde{\mathcal{I}}=\mathcal{I}\}\right\|_1 \\
	& =\sum_{\mathcal{I} \in \mathscr{S}(\mathbb{H})} p_\mathcal{I}\times\left\|\mathbb{H}_\mathcal{I}^{-1} \times\frac{\EE}{n_\textnormal{R}} \mathbb{H}_\mathcal{I} \mathbf{1}_{n_\textnormal{R}}\right\|_1 \\
	& =\sum_{\mathcal{I} \in \mathscr{S}(\mathbb{H})} p_\mathcal{I} \times \EE \\
	& =\EE.
\end{align}
The following theorem provides a lower bound on channel capacity based on the achievable rate induced by this compound exponential distribution.
\begin{thm}[Lower Bound]\label{thm: lb}
\begin{align}
	\const{C} \geq \frac{n_\textnormal{R}}{2} \log \left\{ 1+\biggl(\sum_{\mathcal{I} \in \mathscr{S}(\mathbb{H}) } \bigl|\operatorname{det}\left(\mathbb{H}_\mathcal{I}\right)\bigr|\biggr)^{\frac{2}{n_\textnormal{R}}} \frac{e \EE^2}{2 \pi n_\textnormal{R}^2 \sigma^2}\right\}. \label{eq: theorem C_lb}
\end{align}	
\end{thm}

Second, the following theorem presents two duality-based upper bounds.
\begin{thm}[Upper Bound]\label{thm: ub}
\begin{align}
	\const{C} 
	&\leq \sup _{\mathbf{p}} \inf_{\delta>0, \beta>0} \Biggl\{
	\log \biggl(\sum\nolimits_{\mathcal{I} \in \mathscr{S}(\mathbb{H}) } \left|\operatorname{det}\left(\mathbb{H}_\mathcal{I}\right)\right|\biggr)\nonumber\\
	&\quad -\textsf{D}(\mathbf{p}  \| \mathbf{q} )-\frac{n_\textnormal{R}}{2} \log \left(2 \pi e \sigma^2\right)+\frac{n_\textnormal{R} \delta+\EE}{\beta} \nonumber\\
	&\quad+\sum\limits_{\mathcal{I} \in \mathscr{S}(\mathbb{H}) } p_\mathcal{I} \sum_{l=1}^{n_\textnormal{R}} \Bigg[ \log \biggl(\beta e^{-\frac{\delta^2}{2 \bar{\sigma}_{\mathcal{I}, l}^2}}+\sqrt{2 \pi} \bar{\sigma}_{\mathcal{I}, l} Q\Bigl(\frac{\delta}{\bar{\sigma}_{\mathcal{I}, l}}\Bigr)\biggr) \nonumber\\
	&\quad+\frac{1}{2} Q\left(\frac{\delta}{\bar{\sigma}_{\mathcal{I}, l}}\right) + \frac{\delta}{2 \sqrt{2 \pi} \bar{\sigma}_{\mathcal{I}, l}} e^{-\frac{\delta^2}{2 \bar{\sigma}_{\mathcal{I}, l}^2}}+\frac{\delta^2}{2 \bar{\sigma}_{\mathcal{I}, l}^2}  \nonumber\\
	&\quad+ \frac{\bar{\sigma}_{\mathcal{I}, l}}{\sqrt{2 \pi} \beta} e^{-\frac{\delta^2}{2 \bar{\sigma}_{\mathcal{I}, l}^2}}
	\Biggr] \Biggr\}, \label{eq: thm ub1}\\
	\const{C} 
	&\leq \sup _{\mathbf{p}} \inf_{\nu>0} \Biggl\{
	\log \biggl(\sum\nolimits_{\mathcal{I} \in \mathscr{S}(\mathbb{H}) } \left|\operatorname{det}\left(\mathbb{H}_\mathcal{I}\right)\right|\biggr)\nonumber\\
	&\quad -\textsf{D}(\mathbf{p}  \| \mathbf{q} )-\frac{n_\textnormal{R}}{2} \log \left(2 \pi e \sigma^2\right)+\frac{\EE}{\nu} \nonumber\\
	&\quad+\sum\limits_{\mathcal{I} \in \mathscr{S}(\mathbb{H}) } p_\mathcal{I} \sum_{l=1}^{n_\textnormal{R}} \Bigg[ \frac{\bar{\sigma}_{\mathcal{I}, l}}{\sqrt{2\pi}\nu} 
	+ \log\left(\nu + \sqrt{\frac{\pi e}{2}} \bar{\sigma}_{\mathcal{I}, l}\right)
	\Biggr] \Biggr\}, \label{eq: thm ub2}
\end{align}
where $\mathbf{p}$ and $\mathbf{q}$ denote probability vectors in \eqref{eq: p pdf} and \eqref{eq: q vec def}, respectively, and $\bar{\sigma}_{\mathcal{I},l}$ is given by
\begin{align}
	\bar{\sigma}_{\mathcal{I}, l}=\sigma \sqrt{\left[\mathbb{H}_\mathcal{I}^{-1} \mathbb{H}_\mathcal{I}^{-\textsf{T}}\right]_{l, l}}, \quad l\in\{1,\cdots,n_\textnormal{R}\},
\end{align}
with $[\mathbb{H}_\mathcal{I}^{-1} \mathbb{H}_\mathcal{I}^{-\textsf{T}}]_{l, l}$ denoting the $(l,l)$-th entry of $\mathbb{H}_\mathcal{I}^{-1} \mathbb{H}_\mathcal{I}^{-\textsf{T}}$.
\end{thm}

Third, by combining Theorems~\ref{thm: lb} and \ref{thm: ub}, we obtain the asymptotic capacity in the high-SNR regime.
\begin{thm}[High-SNR Asymptotics]\label{thm: hsnr}
\begin{align}
	&\lim _{\frac{\EE}{\sigma} \rightarrow+\infty}\left\{\const{C}-n_\textnormal{R} \log \left(\frac{\EE}{\sigma}\right)\right\}\nonumber\\
	&\qquad=\log \biggl(\sum_{\mathcal{I}\in\mathscr{S}(\mathbb{H})}\left|\operatorname{det}\left(\mathbb{H}_\mathcal{I}\right)\right|\biggr)+\frac{n_\textnormal{R}}{2} \log \left(\frac{e}{2 \pi n_\textnormal{R}^2}\right). \label{eq: hsnr capacity}
\end{align}
\end{thm}
\begin{rmk}
The term $|\operatorname{det}\left(\mathbb{H}_\mathcal{I}\right)|$ measures the effective volume spanned by the subset of transmit apertures indexed by $\mathcal{I}$ in the receive signal space. Hence, $\sum_{\mathcal{I}\in\mathscr{S}(\mathbb{H})}\left|\operatorname{det}\left(\mathbb{H}_\mathcal{I}\right)\right|$ aggregates the geometric contributions of all active aperture subsets that can support minimum-$\ell_1$-norm image-domain representations. The high-SNR power offset is therefore determined not only by the channel gains but also by the geometry of the image cone induced by the placement of transmit and receive apertures. As will be shown later, for the case $n_\textnormal{T}<n_\textnormal{R}$, the corresponding geometric contribution becomes $|\operatorname{det}(\mathbb{H}^\textsf{T} \mathbb{H})|^{\frac{1}{2}}$, i.e., the square root of the Gram determinant of the channel vectors in the receive signal space.
\hfill $\blacklozenge$
\end{rmk}
\section{Extensions to $n_\textnormal{T}< n_\textnormal{R}$}\label{sec: nt<nr}
This section also employs the NN-BP result to characterize the channel capacity for $n_\textnormal{T}< n_\textnormal{R}$. To begin with, we apply singular value decomposition (SVD) to $\mathbb{H}$:
\begin{align}
	\mathbb{H}=\mathbb{U} \mathbb{S} \mathbb{V}^\textsf{T},
\end{align}
where $\mathbb{S}$ is an $n_\textnormal{R}\times n_\textnormal{T}$ rectangular diagonal matrix, $\mathbb{U}$ and $\mathbb{V}$ are orthogonal matrices. Since $\mathbb{H}$ has full column rank, it follows that $\mathbb{S}$ also has full column rank. For convenience, we denote the first $n_\textnormal{T}$ rows of $\mathbb{S}$ as $\mathbb{S}_{1:n_\textnormal{T}}$. Based on the SVD, the mutual information can be simplified as 
\begin{IEEEeqnarray}{rCl}
	\const{I}( \mathbf{X};\mathbf{Y} ) 
	&=& \const{I}( \mathbf{X}; \mathbb{H} \mathbf{X} + \mathbf{Z}) \\
	&=& \const{I}( \mathbf{X}; \mathbb{U} \mathbb{S} \mathbb{V}^\textsf{T} \mathbf{X} + \mathbf{Z}) \\
	&=& \const{I} \left( \mathbf{X}; \mathbb{S} \mathbb{V}^\textsf{T} \mathbf{X} + \mathbf{Z}^\prime \right) \\
	&=& \const{I} \left( \mathbf{X}; \left[\begin{array}{l}
		\mathbb{S}_{1:n_\textnormal{T}}\\
		\mathbf{0}_{(n_\textnormal{R}-n_\textnormal{T})\times n_\textnormal{T}}
	\end{array} \right] \mathbb{V}^\textsf{T} \mathbf{X} + \left[\begin{array}{l}
	\mathbf{Z}_{1:n_\textnormal{T}}^\prime\\
	\mathbf{Z}_{n_\textnormal{T}+1:n_\textnormal{R}}^\prime
	\end{array} \right] \right) \quad\quad\\
	&=& \const{I} \left( \mathbf{X}; \mathbb{S}_{1:n_\textnormal{T}} \mathbb{V}^\textsf{T} \mathbf{X} + \mathbf{Z}_{1:n_\textnormal{T}}^\prime \right) \\
	&=& \const{I} \left( \mathbf{X}; \widetilde{\mathbb{H}} \mathbf{X} + \widetilde{\mathbf{Z}}  \right), \label{eq: equal channel}
\end{IEEEeqnarray}
where $\mathbf{Z}^\prime=\mathbb{U}^\textsf{T} \mathbf{Z} $, $\mathbf{Z}_{1:n_\textnormal{T}}^\prime$ and $\mathbf{Z}_{n_\textnormal{T}+1:n_\textnormal{R}}^\prime$ are the first $n_\textnormal{T}$ and last $n_\textnormal{R}-n_\textnormal{T}$ components of $\mathbf{Z}^\prime$, respectively. Moreover, we denote $\widetilde{\mathbb{H}}=\mathbb{S}_{1:n_\textnormal{T}} \mathbb{V}^\textsf{T}$ and $\widetilde{\mathbf{Z}}=\mathbf{Z}_{1:n_\textnormal{T}}^\prime$. It can be shown that $\widetilde{\mathbb{H}}$ is invertible and $\widetilde{\mathbf{Z}}$ follows $\mathcal{N}(\mathbf{0}_{n_\textnormal{T}},\sigma^2\mathbb{I}_{n_\textnormal{T}})$.

The above transformation reveals that the original  $n_\textnormal{R}\times n_\textnormal{T}$ channel model $\mathbf{Y}=\mathbb{H}\mathbf{X}+\mathbf{Z}$ is equivalent to an $n_\textnormal{T}\times n_\textnormal{T}$ channel model $\widetilde{\mathbf{Y}}=\widetilde{\mathbb{H}}\mathbf{X}+\widetilde{\mathbf{Z}}$ in \eqref{eq: equal channel}. Hence, it suffices to analyze the capacity of this equivalent $n_\textnormal{T}\times n_\textnormal{T}$ channel. By substituting  $\widetilde{\mathbb{H}}$ for $\mathbb{H}$ in Theorems~\ref{thm: lb} and \ref{thm: ub}, we obtain the following results with proofs given in Appendix~\ref{app: proof of case ii}.
\begin{thm}\label{thm: case II}
For $n_\textnormal{T}<n_\textnormal{R}$, the channel capacity can be bounded as
	\begin{align}
		\const{C} 
		&\geq \frac{n_\textnormal{T}}{2} \log \left\{ 1+ \bigl|\operatorname{det}(\mathbb{H}^\textsf{T}\mathbb{H})\bigr|^{\frac{1}{n_\textnormal{T}}} \frac{e \EE^2}{2 \pi n_\textnormal{T}^2 \widetilde{\sigma}^2}\right\}, \label{}	\\
		\const{C} 
		&\leq \inf_{\delta>0, \beta>0} \Biggl\{
		\frac{1}{2}\log (|\operatorname{det}(\mathbb{H}^\textsf{T}\mathbb{H})|) - \frac{n_\textnormal{T}}{2} \log \left(2 \pi e \sigma^2\right) \nonumber\\
		&\quad +\frac{n_\textnormal{T} \delta+\EE}{\beta} + \sum_{l=1}^{n_\textnormal{T}} \Bigg[ \log \biggl(\beta e^{-\frac{\delta^2}{2 \widetilde{\sigma}_{l}^2}}+\sqrt{2 \pi} \widetilde{\sigma}_{l} Q\Bigl(\frac{\delta}{\widetilde{\sigma}_{l}}\Bigr)\biggr)  \nonumber\\
		&\quad +\frac{1}{2} Q\left(\frac{\delta}{\widetilde{\sigma}_{l}}\right) + \frac{\delta}{2 \sqrt{2 \pi} \widetilde{\sigma}_{l}} e^{-\frac{\delta^2}{2 \widetilde{\sigma}_{l}^2}}+\frac{\delta^2}{2 \widetilde{\sigma}_{l}^2} \nonumber\\
		&\quad + \frac{\widetilde{\sigma}_{l}}{\sqrt{2 \pi} \beta} e^{-\frac{\delta^2}{2 \widetilde{\sigma}_{l}^2}}
		\Biggr] \Biggr\}, \label{}\\
		\const{C} 
		&\leq \inf_{\nu>0} \Biggl\{
		\frac{1}{2} \log (|\operatorname{det}(\mathbb{H}^\textsf{T}\mathbb{H})|) -\frac{n_\textnormal{T}}{2} \log \left(2 \pi e \sigma^2\right) +\frac{\EE}{\nu} \nonumber\\
		&\quad+ \sum_{l=1}^{n_\textnormal{T}} \Bigg[ \frac{\widetilde{\sigma}_{l}}{\sqrt{2\pi}\nu} 
		+ \log\left(\nu + \sqrt{\frac{\pi e}{2}} \widetilde{\sigma}_{l}\right)
		\Biggr] \Biggr\}, \label{}
	\end{align}
	where $\widetilde{\sigma}_{l}$ is given by
	\begin{align}
		\widetilde{\sigma}_{l}=\sigma \sqrt{[( \mathbb{H}^{\textsf{T}} \mathbb{H})^{-1}]_{l, l}}, \quad l\in\{1,\cdots,n_\textnormal{T}\},
	\end{align}
	with $[( \mathbb{H}^{\textsf{T}} \mathbb{H})^{-1}]_{l, l}$ denoting the $(l,l)$-th entry of $( \mathbb{H}^{\textsf{T}} \mathbb{H})^{-1}$. 
	
	At high SNR, the asymptotic capacity is given by
	\begin{align}
		&\lim _{\frac{\EE}{\sigma} \rightarrow+\infty}\left\{\const{C}-n_\textnormal{T} \log \left(\frac{\EE}{\sigma}\right)\right\}\nonumber\\
		&\qquad=\frac{1}{2}\log (|\operatorname{det}(\mathbb{H}^\textsf{T} \mathbb{H})|)+\frac{n_\textnormal{T}}{2} \log \left(\frac{e}{2 \pi n_\textnormal{T}^2}\right). \label{}
	\end{align}
\end{thm}

\section{Numerical Results}\label{sec: numerical results}
\begin{table}[t]
	\centering
	\caption{Indoor VLC system parameters}
	\label{tab:rx_channel}
	\renewcommand{\arraystretch}{1}
	\setlength{\tabcolsep}{5pt}
	\begin{tabular}{c|c}
		\hline
		\rule[-1pt]{0pt}{10pt} Parameter & Value \\
		\hline
		Rx position
		&
		\begin{tabular}{@{}c@{}}
			PD 1: $(0.6776, 0.25, 0.75)\,\mathrm{m}$ \\
			PD 2: $(0.5044, 0.15, 0.75)\,\mathrm{m}$
		\end{tabular}
		\\
		\hline
		Channel matrix
		&
		\begin{tabular}{@{}c@{}}
			$\mathbb{H}_\textnormal{a}={\scriptsize 10^{-6} 
				\begin{pmatrix}
					4.90983 & 3.55895 & 2.26255 & 1.80856 \\
					4.43089 & 3.66528 & 2.48780 & 2.15458
			\end{pmatrix}}$
		\end{tabular}
		\\
		\hline
	\end{tabular}
\end{table}

\begin{figure}[t]
	\centering
	\includegraphics[width=2.5in]{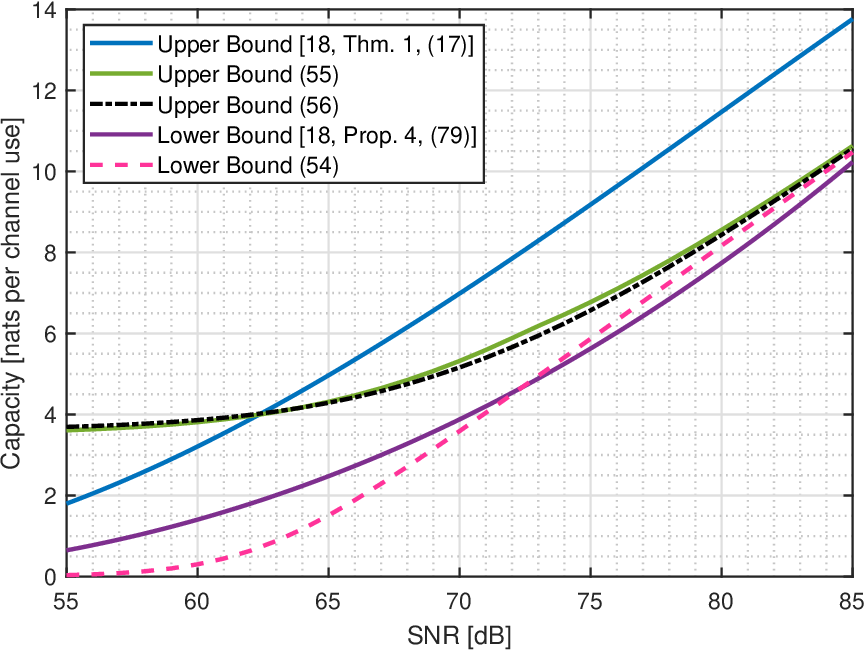}
	\caption{Capacity bounds of a $2\times4$ VLC channel with $\mathbb{H}=\mathbb{H}_\textnormal{a}$.}
	\label{fig 2x4 U1}
\end{figure}
\begin{figure}[t]
	\centering
	\vspace{-2mm}
	\subfloat[{$\mathscr{S}(\mathbb{H}_\textnormal{b})=\{1,3\}$.}]{
		\includegraphics[width=2.5in]{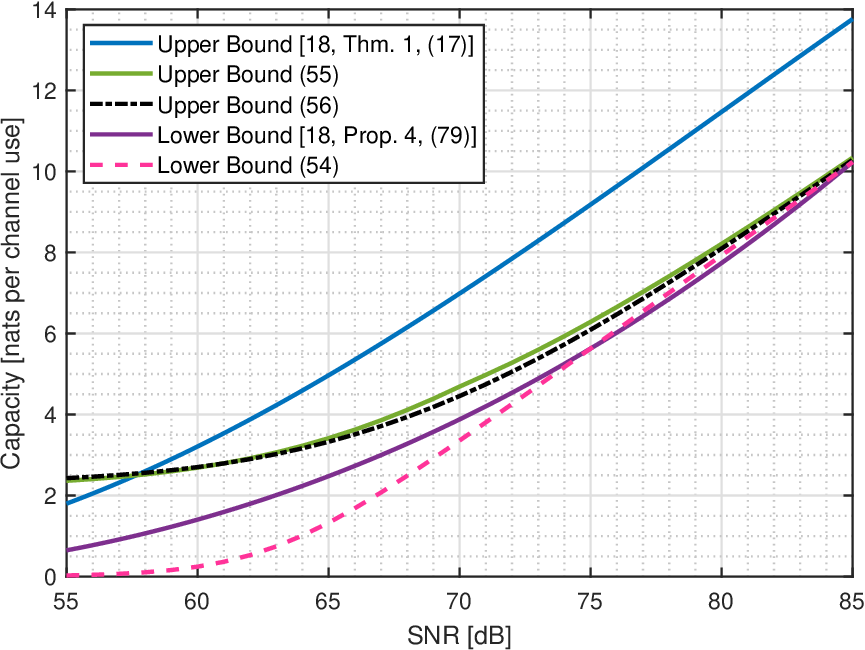}\label{fig:sub1}
	}\hspace{-2mm}	
	\subfloat[{$\mathscr{S}(\mathbb{H}_\textnormal{b})=\{\{1,2\},\{2,3\}\}$.}]{
		\includegraphics[width=2.5in]{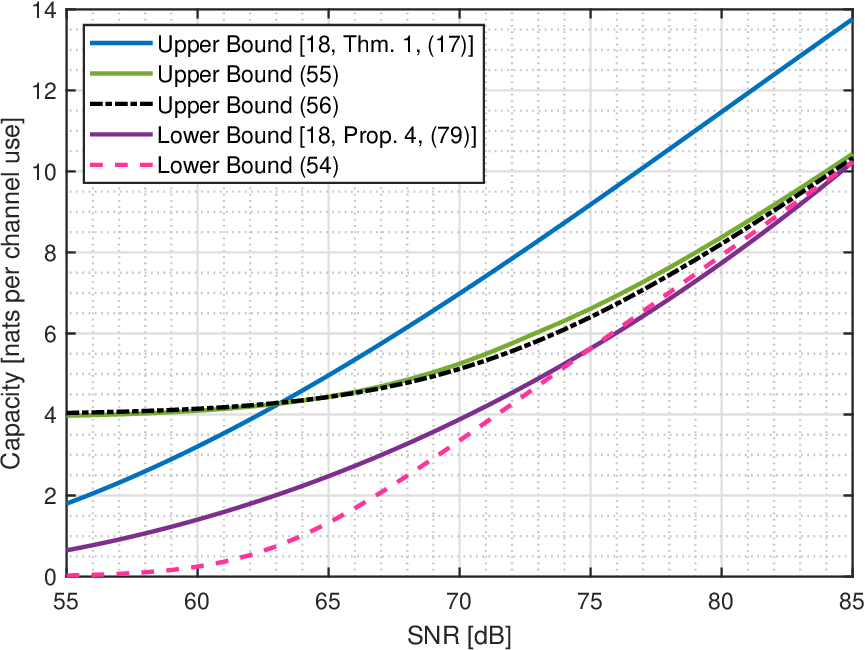}\label{fig:sub2}
	}
	\caption{Capacity bounds of a $2\times3$ VLC channel with $\mathbb{H}=\mathbb{H}_\textnormal{b}$.}
	\label{fig 2x3 U2}
\end{figure}



This section provides numerical examples to illustrate the proposed channel capacity bounds. We first consider an indoor VLC scenario, followed by an outdoor fading FSO scenario.


\subsection{Indoor VLC Scenario}
We start with an indoor VLC system consisting of four LEDs and two PDs deployed in a $4.0\,\mathrm{m} \times 4.0\,\mathrm{m} \times 3.0\,\mathrm{m}$ room. The LEDs are placed on the ceiling and point vertically downward, whereas the PDs are placed at a height of $0.75\,\mathrm{m}$ (a typical table height), and point vertically upward toward the ceiling. The four LEDs are located at $(\pm 1 , \pm 1, 3)\,$m, and the two PDs are separated by a distance of $0.2\,\mathrm{m}$. Since optical wireless links are typically dominated by the line-of-sight (LOS) component, we consider only the LOS channel gain, which is calculated according to the Lambertian model \cite[Eq. (3)]{Fath2013}. The Lambertian order of the LED is set to 1, the field-of-view semi-angle and active area of the PDs are set to $60^\circ$ and 1 $\mathrm{cm}^2$, respectively. The receiver location generated under the above geometric constraints and the resulting channel matrix are summarized in Table~\ref{tab:rx_channel}.

\begin{figure}[t]
	\centering
	\subfloat[$\mathbb{H}=\mathbb{H}_\textnormal{a}(1,:)$.]{
		\includegraphics[width=1.65in]{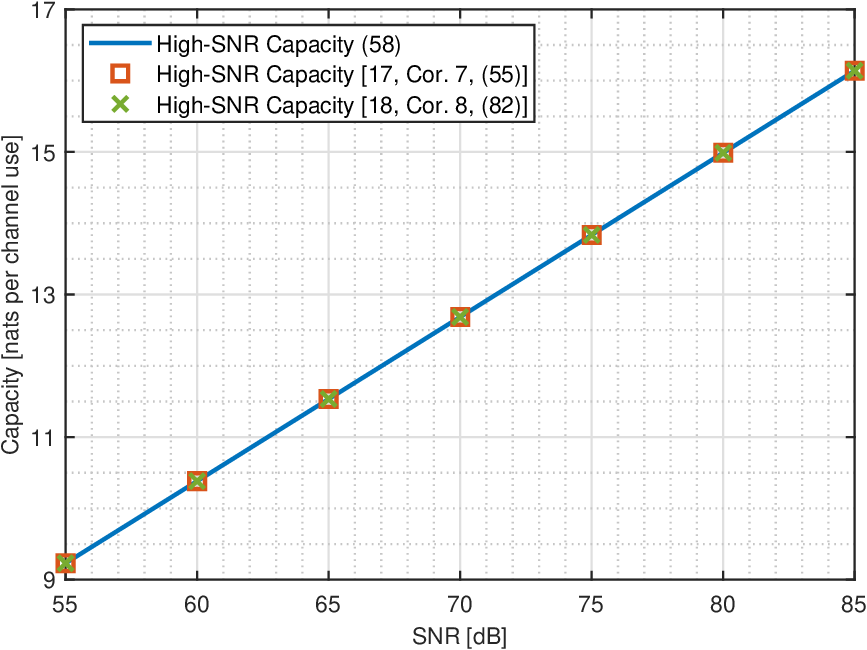}\label{}
	}\hspace{-2mm}	
	\subfloat[$\mathbb{H}=\mathbb{H}_\textnormal{a}(2,:)$.]{
		\includegraphics[width=1.65in]{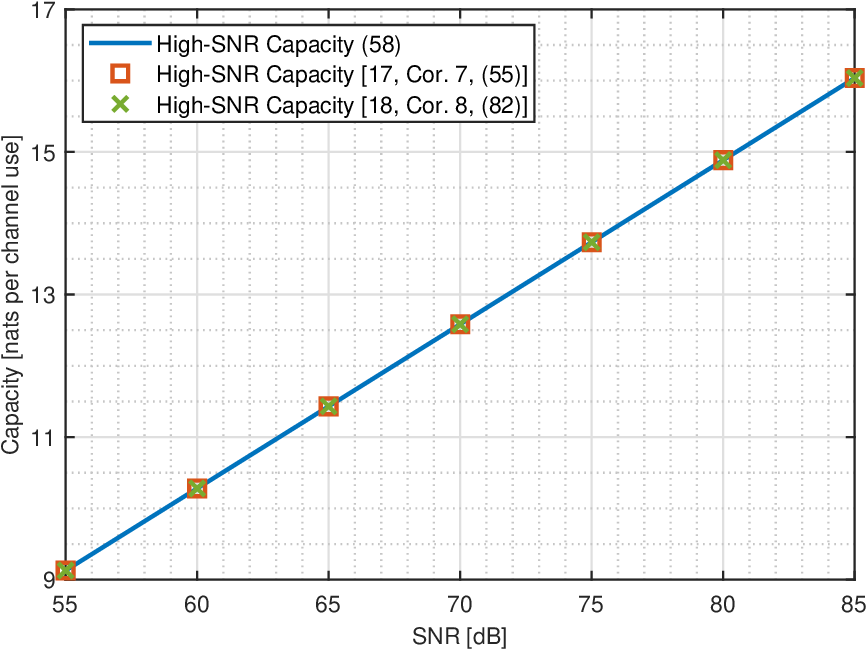}\label{2}
	}
	\caption{High-SNR capacity for $1\times4$ VLC channels.}
	\label{fig: high snr case I}
	\vspace{-3mm}
\end{figure}

Fig.~\ref{fig 2x4 U1} illustrates the proposed lower and upper bounds for the VLC channel. For comparison, the existing bounds in \cite{Chaaban2018-hsnr} are also included. It can be observed that the proposed lower and upper bounds become increasingly tight and eventually coincide as the SNR increases, which confirms their asymptotic optimality in the high-SNR regime. In contrast, the bounds in \cite{Chaaban2018-hsnr} are much looser than the proposed bounds. Moreover, the gap between the upper and lower bounds in \cite{Chaaban2018-hsnr} does not necessarily decrease monotonically with the SNR. 

For Fig.~\ref{fig 2x4 U1}, substituting the corresponding channel matrix into Lemma~\ref{lem 1} shows that the associated set $\mathscr{S}(\mathbb{H})$ is unique. Specifically, $\mathscr{S}(\mathbb{H}_\textnormal{a})=\{\{1,2\},\{2,4\}\}$. To further demonstrate the situation where $\mathscr{S}(\mathbb{H})$ is not unique, we turn off LED 2 and construct a new channel matrix $\mathbb{H}_\textnormal{b}$ by extracting the 1st, 3rd and 4th columns of $\mathbb{H}_\textnormal{a}$. Here, the corresponding $\mathscr{S}(\mathbb{H}_\textnormal{b})$ is not unique. Two possible choices are given by $\{1,3\}$ and $\{\{1,2\},\{2,3\}\}$. We therefore evaluate the proposed bounds under these two choices and present the corresponding results in Figs. \eqref{fig:sub1} and \eqref{fig:sub2}, respectively. It can be observed that, regardless of the choice of $\mathscr{S}(\mathbb{H}_\textnormal{b})$, the proposed bounds remain tight and asymptotically coincide as the SNR increases, thereby validating the correctness of our derivations.




We also plot the high-SNR asymptotic capacity in Fig.~\ref{fig: high snr case I}. It can be observed that all curves exhibit the same growth rate with respect to the SNR, which confirms the proposed high-SNR capacity characterization. It is worth noting that the existing results in \cite{Chaaban2018-hsnr,Moser2017-ISIT} are applicable only to the special case of $n_\textnormal{R} = 1$. In contrast, this paper characterizes the channel capacity for general MIMO configurations.


\begin{figure}[t]
	\centering
	\includegraphics[width=2.5in]{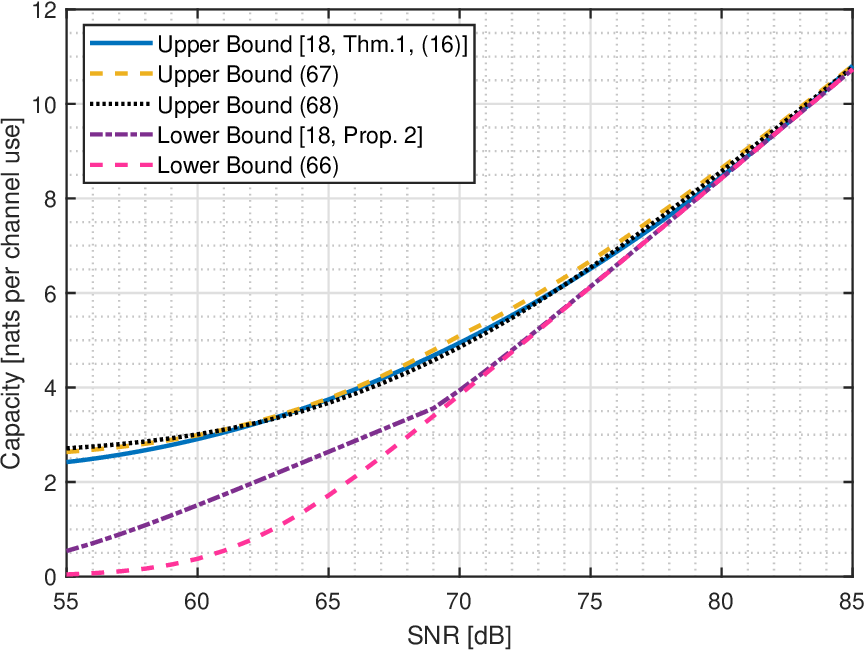}
	\caption{Capacity bounds of a $4\times2$ VLC channel with $\mathbb{H}=\mathbb{H}_\textnormal{a}^\textsf{T}$.}
	\label{fig 2x4 U1 case II}
	\vspace{-3mm}
\end{figure}

\begin{figure}[t]
	\centering
	\includegraphics[width=2.5in]{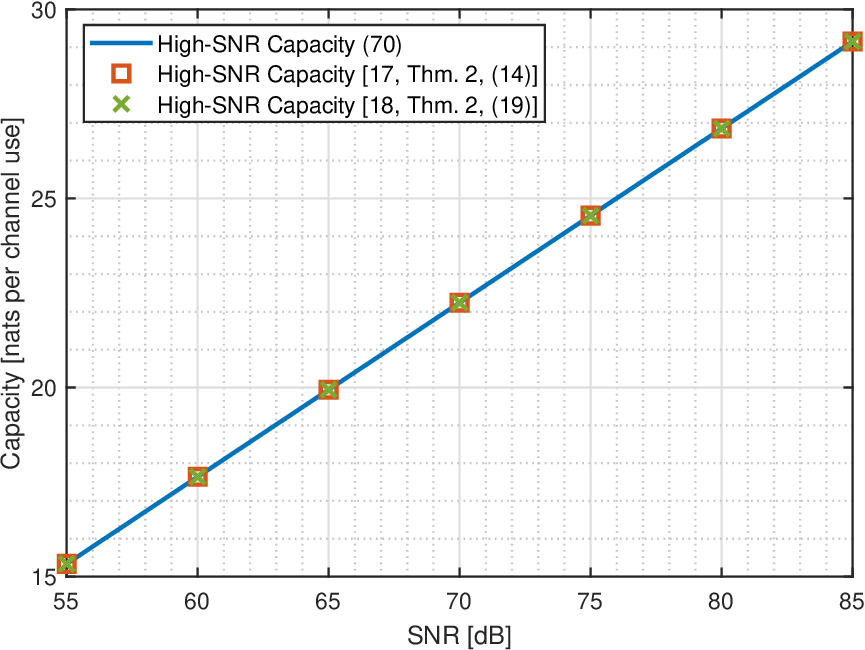}
	\caption{High-SNR capacity for a $4\times2$ VLC channel with $\mathbb{H}=\mathbb{H}_\textnormal{a}^\textsf{T}$.}
	\label{fig high SNR case II}
	\vspace{-3mm}
\end{figure}

For the case $n_\textnormal{T}< n_\textnormal{R}$, we set $\mathbb{H}=\mathbb{H}_\textnormal{a}^\textsf{T}$ for simplicity. Fig.~\ref{fig 2x4 U1 case II} depicts the proposed lower and upper bounds, in comparison with the existing bounds from \cite{Chaaban2018-hsnr}. Meanwhile, Fig.~\ref{fig high SNR case II} plots the high-SNR asymptotic capacity alongside the corresponding results in \cite{Chaaban2018-hsnr,Moser2017-ISIT}. Similar to the case $n_\textnormal{T}\geq n_\textnormal{R}$, the proposed lower and upper bounds become increasingly tight as the SNR increases, and the high-SNR capacity exhibits the same growth rate as the existing results.


\begin{table}[b]
	\centering
	\caption{Outdoor FSO system parameters}
	\label{tab:simulation_parameters}
	\begin{tabular}{lll}
		\hline
		Parameter & Symbol & Value \\
		\hline	
		Deterministic path loss & $h_\ell$ & 0.008\\
		Rytov variance & $\sigma_\textnormal{R}^2$ & 0.1\\
		Receiver diameter & $2a$ & 20 cm \\
		Corresponding Beam radius at 1 km & $w_z$ & $2.5$ m \\
		Jitter standard deviation & $\sigma_s$ & $30$ cm \\
		\hline
	\end{tabular}
\end{table}

\subsection{Outdoor Fading FSO Channel}
We next evaluate the proposed bounds in an outdoor FSO system, where the channel undergoes fading due to atmospheric turbulence and pointing error. To focus on the impact of random channel fading on the capacity bounds, the entries $h_{ij}$ of the channel matrix are assumed to be independent and identically distributed. Following the FSO channel model in \cite{Farid2007}, each $h_{ij}$ is determined by the product of three components: the deterministic path-loss term $h_\ell$, the random atmospheric-turbulence fading coefficient $h_\textnormal{a}$, and the random pointing-error fading coefficient $h_\textnormal{p}$. Specifically, $h_\textnormal{a}$ follows a log-normal distribution, whereas $h_\textnormal{p}$ follows a bounded power-law distribution characterized by the Gaussian beam profile at the receiver plane and the Rayleigh-distributed radial displacement.

We consider a light-fog scenario with a propagation distance of 1 km. The system parameters are summarized in Table~\ref{tab:simulation_parameters} and substituted into \cite[Eq. (14)]{Farid2007} to obtain the distribution of $h_{ij}$. Using this distribution, we generate independent FSO channel realizations and evaluate the proposed lower and upper bounds for each realization. The ergodic capacity bounds are then obtained by averaging over all channel realizations. Figs.~\ref{fig 2x3 ergodic capacity} and \ref{fig 3x2 ergodic capacity} present the ergodic capacity results for the FSO channels. Similar to the indoor VLC results, the proposed bounds remain tighter than the existing bounds in \cite{Chaaban2018-hsnr}. Moreover, the gap between the proposed lower and upper bounds decreases with the SNR, and the two bounds become asymptotically tight in the high-SNR regime.


\begin{figure}[t]
	\centering
	\includegraphics[width=2.5in]{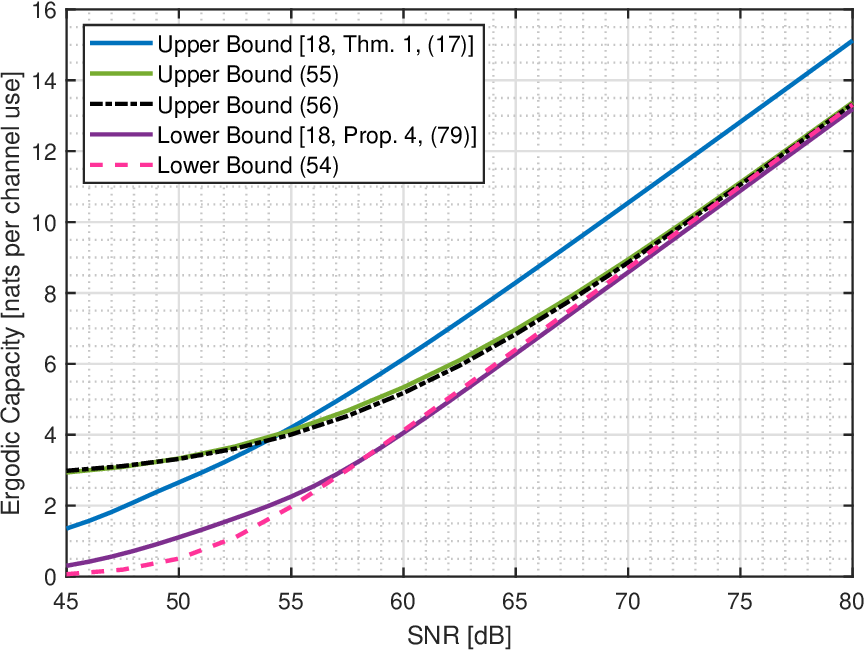}
	\caption{Ergodic capacity bounds of a $2\times4$ FSO channel.}
	\label{fig 2x3 ergodic capacity}
	\vspace{-3mm}
\end{figure}

\begin{figure}[t]
	\centering
	\includegraphics[width=2.5in]{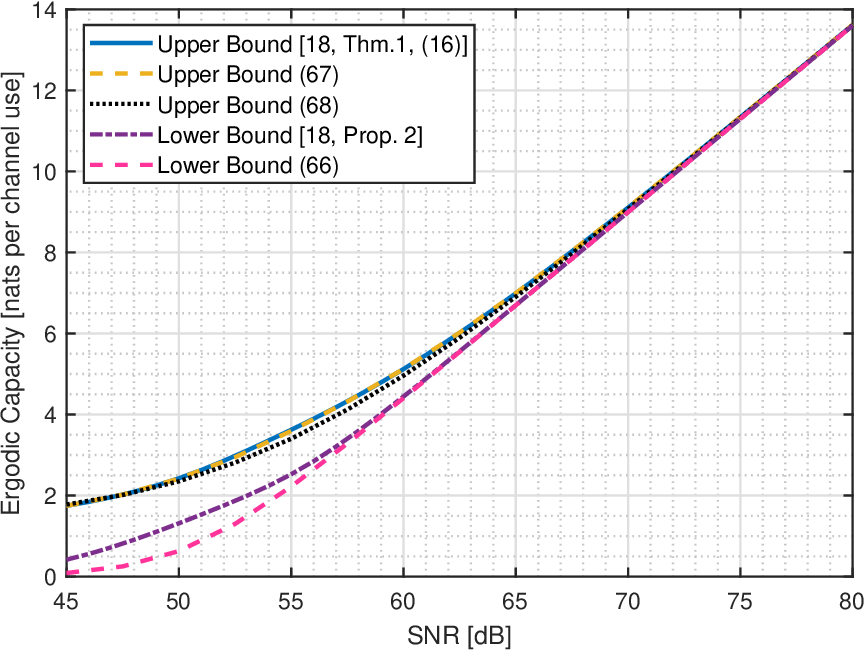}
	\caption{Ergodic capacity bounds of a $4\times2$ FSO channel.}
	\label{fig 3x2 ergodic capacity}
	\vspace{-3mm}
\end{figure}

\section{Conclusions} \label{sec:conclusion}
In this paper, we characterize the capacity of the MIMO-OWC channels under a total average-power constraint. Two typical aperture configurations are investigated: one in which the number of transmit apertures is greater than or equal to that of receive apertures, and the other in which the number of transmit apertures is smaller than that of receive apertures. By exploiting the Markov chain among the input vector $\mathbf{X}$, the image vector $\bar{\mathbf{X}}$, and the output vector $\mathbf{Y}$, the channel capacity can be reformulated as a maximization problem over the distribution of $\bar{\mathbf{X}}$. The key difficulty in this reformulation lies in translating the total average-power constraint imposed on $\mathbf{X}$ into an equivalent constraint on $\bar{\mathbf{X}}$. To address this issue, we formulate an NN-BP problem, which identifies the optimal input vector with minimum $\ell_1$-norm for any image vector. We show that the cone $\mathcal{S}(\mathbb{H})$, which contains all feasible image vectors, can be partitioned into a collection of subcones $\mathcal{S}(\mathbb{H}_{\mathcal{I}})$. Within each subcone, the NN-BP solution admits a unified form. Based on the NN-BP result, an equivalent capacity expression in terms of the image vector $\bar{\mathbf{X}}$ is derived. We then establish a lower bound on the channel capacity by constructing a compound exponential distribution and upper bounds by carefully designing auxiliary output distributions. The proposed capacity bounds tighten the existing bounds in the literature and are asymptotically optimal in the high-SNR regime, thereby closing the gap left by prior studies. We further extend the results to the case where the number of transmit apertures is smaller than that of receive apertures, for which capacity bounds and the high-SNR asymptotic capacity are also derived. For future research, it would be interesting to extend the analysis to the low-SNR regime and characterize the corresponding capacity slope.

\appendices
\section{Derivation of the NN-BP Solution}\label{app: proof of lem 1}
\subsection{Proof of Lemma~\ref{lem 1}}
To facilitate the proof of Lemma~\ref{lem 1}, we first present the following auxiliary lemma.
\begin{lem}\label{lem6}
If \eqref{eq: gamma_{I,j} neq 1} is satisfied, i.e., $\forall \mathcal{I} \in \mathscr{R}(\mathbb{H})$ and $\forall j \in \mathcal{I}^c$, $\gamma_{\mathcal{I},j} \neq 1$, then $\forall\bar{\mathbf{x}}\in\mathcal{S}(\mathbb{H})$, except a subset of zero $n_\textnormal{R}$-dimensional volume, the optimal solution $\mathbf{x}^\star$ to the NN-BP problem belongs to the set $\cup_{\mathcal{I}\in\mathscr{S}(\mathbb{H})} \mathcal{X}_{\mathcal{I}}$, where 
\begin{IEEEeqnarray}{rCl}
	\mathcal{X}_{\mathcal{I}} =\{&&\mathbf{x}=(x_1,\cdots,x_{n_\textnormal{T}}):\nonumber\\
	&&\quad x_i\in(0,+\infty),i\in\mathcal{I};x_j=0,j\in\mathcal{I}^c\}.
\end{IEEEeqnarray}
\end{lem}
\begin{IEEEproof}
Denote the Lagrangian function for the NN-BP problem by 
\begin{align}
	L(\mathbf{x},\bm{\mu},\bm{\lambda}) = \mathbf{1}_{n_\textnormal{T}}^\textsf{T} \mathbf{x} - \bm{\mu}^\textsf{T} \mathbf{x} + \bm{\lambda}^\textsf{T} (\mathbb{H}\mathbf{x}-\bar{\mathbf{x}}),
\end{align}
where $\bm{\mu}$ and $\bm{\lambda}$ are dual variables in $\mathfrak{R}^{n_\textnormal{T}}$.
By the Karush-Kuhn-Tucker conditions, the NN-BP solution $\mathbf{x}^\star$ satisfies
\begin{IEEEeqnarray}{rCl}
	\subnumberinglabel{}
	\mathbf{1}_{n_\textnormal{T}}  - \bm{\mu}^\star + \mathbb{H}^\textsf{T}\bm{\lambda}^\star &=& \mathbf{0}_{n_\textnormal{T}}, \label{eq: stability}\\	
	\mathbb{H}\mathbf{x}-\bar{\mathbf{x}}&=&\mathbf{0}_{n_\textnormal{R}},\\
	\mathbf{x}^\star&\geq& \mathbf{0}_{n_\textnormal{T}},\\
	\bm{\mu}^\star &\geq&  \mathbf{0}_{n_\textnormal{T}}, \label{eq: dual feasible}\\
	\mu_i^\star x_i^\star &=& 0, \quad\forall i\in\{1,\cdots,n_\textnormal{T}\}, \label{eq: hubu}
\end{IEEEeqnarray}	
where $\bm{\mu}^\star$ and $\bm{\lambda}^\star$ are the optimal dual variables. Define the index set $\mathcal{J}^\star$ as
\begin{IEEEeqnarray}{rCl}
	\mathcal{J}^\star = \{i\in\{1,\cdots,n_\textnormal{T}\}:x_i\in(0,+\infty)\}. \label{eq: H_J def}
\end{IEEEeqnarray}
Note that $\textnormal{Vol}_{n_\textnormal{R}} (\mathcal{S}(\mathbb{H}_{\mathcal{J}^\star}))=0$ if $\textnormal{rank} (\mathbb{H}_{\mathcal{J}^\star}) <n_\textnormal{R}$. Thus, for the purpose of this proof, we assume 
\begin{align}
	\textnormal{rank} (\mathbb{H}_{\mathcal{J}^\star}) =n_\textnormal{R}.
\end{align}

We first prove that there exists an index set $\mathcal{I}^\star\subseteq\mathcal{J}^\star$ that belongs to $\mathscr{S}(\mathbb{H})$. Recall that $\mathscr{R}(\mathbb{H})$ contains all $\mathcal{K}$s such that $\textnormal{rank} (\mathbb{H}_\mathcal{K}) =n_\textnormal{R}$. Hence, there exists $\mathcal{I}^\star\in\mathscr{R}(\mathbb{H})$ such that $\mathcal{I}^\star\subseteq\mathcal{J}^\star$. It remains to show that $\gamma_{\mathcal{I}^\star,j} <1$ for all $j\in(\mathcal{I}^\star)^c$. From \eqref{eq: hubu} and \eqref{eq: H_J def}, we have
\begin{align}
	\mu_i^\star =0,\quad \forall i\in\mathcal{J}^\star. \label{eq: m_istar=0}
\end{align}
By \eqref{eq: stability}, we extract the rows related to $\mathcal{I}^\star$ and $(\mathcal{I}^\star)^c$, and obtain
\begin{align}
	\mathbf{1}_{n_\textnormal{R}} + \bm{\mu}_{\mathcal{I}^\star}^\star + \mathbb{H}_{\mathcal{I}^\star}^\textsf{T}\bm{\lambda}^\star &= \mathbf{0}_{n_\textnormal{R}}, \label{eq: stability I}\\
	1-\mu_j^\star + \mathbf{h}_j^\textsf{T} \bm{\lambda}^\star &=0, \qquad \forall j\in(\mathcal{I}^\star)^c, \label{eq: stability I_c}
\end{align}
where $\bm{\mu}_{\mathcal{I}^\star}^\star=(\mu_{i_1},\cdots,\mu_{i_{n_\textnormal{R}}})^\textsf{T}$ with $\mathcal{I}^\star=\{i_1,\cdots,i_{n_\textnormal{R}}\}$. Substituting \eqref{eq: m_istar=0} into \eqref{eq: stability I}, we have
\begin{align}
	\bm{\lambda}^\star=-\mathbb{H}_{\mathcal{I}^\star}^{-\textsf{T}}\mathbf{1}_{n_\textnormal{R}}. \label{eq: lambda_star}
\end{align}
Combining \eqref{eq: stability I_c} with \eqref{eq: lambda_star} yields
\begin{align}
	1-\mu_j^\star 
	&= \mathbf{h}_j^\textsf{T} \mathbb{H}_{\mathcal{I}^\star}^{-\textsf{T}}\mathbf{1}_{n_\textnormal{R}}\\
	&= \gamma_{\mathcal{I}^\star,j}, \qquad\  \forall j\in(\mathcal{I}^\star)^c, \label{eq: 59}
\end{align}
where \eqref{eq: 59} is derived from the definition in \eqref{eq: gamma def}. By inserting \eqref{eq: dual feasible} into \eqref{eq: 59}, we obtain 
\begin{align}
	\gamma_{\mathcal{I}^\star,j} \leq 1,\quad \forall j\in(\mathcal{I}^\star)^c.
\end{align}
Under condition \eqref{eq: gamma_{I,j} neq 1}, it follows that $\gamma_{\mathcal{I}^\star,j} < 1$ for all $j\in(\mathcal{I}^\star)^c$, which indicates that $\mathcal{I}^\star\in\mathscr{S}(\mathbb{H})$.

We further prove the reverse inclusion, i.e., $\mathcal{I}^\star\supseteq\mathcal{J}^\star$.
To this end, suppose that there exists an element $k\in\mathcal{J}^\star \setminus \mathcal{I}^\star$. From \eqref{eq: m_istar=0}, we have 
\begin{align}
	\mu_k^\star =0.
\end{align}
Note that $k\in(\mathcal{I}^\star)^c$. Then, with \eqref{eq: stability I_c}, we have 
\begin{align}
	1 + \mathbf{h}_k^\textsf{T} \bm{\lambda}^\star &=0.
\end{align}
Substituting \eqref{eq: lambda_star} into it, we have
\begin{align}
	1 - \mathbf{h}_k^\textsf{T} \mathbb{H}_{\mathcal{I}^\star}^{-\textsf{T}}\mathbf{1}_{n_\textnormal{R}}=0,
\end{align}
which implies that $\gamma_{\mathcal{I}^\star,k} = 1$, contradicting condition \eqref{eq: gamma_{I,j} neq 1}. This contradiction shows that no such $k$ exists, and hence $\mathcal{I}^\star\supseteq\mathcal{J}^\star$.

Combining $\mathcal{I}^\star \subseteq \mathcal{J}^\star$ and $\mathcal{I}^\star \supseteq \mathcal{J}^\star$, we obtain $\mathcal{I}^\star=\mathcal{J}^\star$. Therefore, under condition \eqref{eq: gamma_{I,j} neq 1}, for any $ \bar{\mathbf{x}}\in\mathcal{S}(\mathbb{H})$, there exists an index set $\mathcal{I}^\star\in\mathscr{S}(\mathbb{H})$ such that the NN-BP solution $\mathbf{x}^\star$ lies in $\mathcal{X}_{\mathcal{I}^\star}$, which completes the proof.
\end{IEEEproof}

We now proceed to prove Lemma~\ref{lem 1}. For the first part of Lemma~\ref{lem 1}, we first consider the case where condition \eqref{eq: gamma_{I,j} neq 1} is satisfied. By Lemma~\ref{lem6} and Berge's maximum theorem \cite{Berge1997}, the NN-BP solution $\mathbf{x}^\star$ to \eqref{eq: minimum-energy problem} is contained in the closure $\operatorname{cl}\left(\cup_{\mathcal{I}\in\mathscr{S}(\mathbb{H})} \mathcal{X}_{\mathcal{I}}\right)$. Then, we can obtain  
\begin{align}
	\mathcal{S}(\mathbb{H})  \subseteq \Biggl\{ \mathbb{H} \mathbf{x}: \, \mathbf{x} \in \operatorname{cl}\Biggl(\bigcup\nolimits_{\mathcal{I}\in\mathscr{S}(\mathbb{H})} \mathcal{X}_{\mathcal{I}}\Biggr) \Biggr\}. \label{eq: subset}
\end{align}
Besides, note that $\operatorname{cl}\left(\cup_{\mathcal{I}\in\mathscr{S}(\mathbb{H})} \mathcal{X}_{\mathcal{I}}\right)\subseteq\mathfrak{R}_+^{n_\textnormal{T}}$, we have
\begin{align}
	\mathcal{S}(\mathbb{H})  \supseteq \Biggl\{ \mathbb{H} \mathbf{x}: \, \mathbf{x} \in \operatorname{cl}\Biggl(\bigcup\nolimits_{\mathcal{I}\in\mathscr{S}(\mathbb{H})} \mathcal{X}_{\mathcal{I}}\Biggr) \Biggr\}. \label{eq: supset}
\end{align}
Combining \eqref{eq: subset} and \eqref{eq: supset} gives
\begin{align}
	\mathcal{S}(\mathbb{H})  = \Biggl\{ \mathbb{H} \mathbf{x}: \, \mathbf{x} \in \operatorname{cl}\Biggl(\bigcup\nolimits_{\mathcal{I}\in\mathscr{S}(\mathbb{H})} \mathcal{X}_{\mathcal{I}}\Biggr) \Biggr\}.
\end{align}
Since the set $\mathscr{S}(\mathbb{H})$ is finite, we further obtain that $\operatorname{cl}( \cup_{\mathcal{I}\in\mathscr{S}(\mathbb{H})} \mathcal{X}_{\mathcal{I}}) =\cup_{\mathcal{I}\in\mathscr{S}(\mathbb{H})} \operatorname{cl}(\mathcal{X}_{\mathcal{I}})$, then
\begin{align}
	\mathcal{S}(\mathbb{H})  
	&= \Bigl\{ \mathbb{H} \mathbf{x}: \, \mathbf{x} \in \bigcup\nolimits_{\mathcal{I}\in\mathscr{S}(\mathbb{H})} \operatorname{cl}\bigl( \mathcal{X}_{\mathcal{I}}\bigr) \Bigr\}\\
	&= \bigcup\nolimits_{\mathcal{I}\in\mathscr{S}(\mathbb{H})} \mathcal{S}(\mathbb{H}_{\mathcal{I}}), \label{eq: 68}
\end{align}
where \eqref{eq: 68} holds since
\begin{align}
	\mathcal{S}(\mathbb{H}_{\mathcal{I}})=\{\mathbb{H}\mathbf{x}:\mathbf{x}\in\operatorname{cl}(\mathcal{X}_{\mathcal{I}}) \}. \label{eq: 67}
\end{align}
By \eqref{eq: 68}, it concludes \eqref{eq: lem5-1}. 

Furthermore, we show that \eqref{eq: lem5-2} holds true under condition~\eqref{eq: gamma_{I,j} neq 1}. For notational convenience, we denote
\begin{align}
	\mathcal{S}_\textnormal{unit} (\mathbb{M}) = \Bigl\{ \sum\nolimits_{i=1}^{s} a_i \mathbf{m}_i:\,\,	&a_1,\cdots,a_s\in[0,1],\nonumber\\
	&\sum\nolimits_{i=1}^{s} a_i \leq 1\Bigr\},
\end{align}
for any $r\times s$ matrix $\mathbb{M}=\left(\mathbf{m}_1,\cdots,\mathbf{m}_s\right)$ with full row rank. We then invoke the following lemma, whose proof is provided in Appendix~\ref{app: proof of lem7}.
\begin{lem}\label{lem7}
If condition \eqref{eq: gamma_{I,j} neq 1} is satisfied, then 
\begin{align}
	&\textnormal{Vol}_{n_\textnormal{R}} \left\{ \mathcal{S}_\textnormal{unit}(\mathbb{H}_\mathcal{I}) \cap \mathcal{S}_\textnormal{unit}(\mathbb{H}_\mathcal{J})\right\} =0, \nonumber\\
	&\quad\qquad\qquad\qquad\qquad\, \forall\mathcal{I},\mathcal{J}\in\mathscr{S}(\mathbb{H}),\, \mathcal{I}\neq\mathcal{J}. \label{eq: 71}
\end{align}
\end{lem}

Since any vector in $\mathcal{S}(\mathbb{H})$ can be normalized to $\mathcal{S}_\textnormal{unit}(\mathbb{H})$, the geometric relationships within $\mathcal{S}(\mathbb{H})$ are preserved in $\mathcal{S}_\textnormal{unit}(\mathbb{H})$. Therefore, we can obtain \eqref{eq: lem5-2} from \eqref{eq: 71}.

We now turn to the proof of the second part of Lemma~\ref{lem 1}, i.e., \eqref{eq: opt x express}. We first consider the case where condition \eqref{eq: gamma_{I,j} neq 1} is satisfied. From \eqref{eq: lem5-2}, we obtain 
\begin{align}
	\operatorname{int} (\mathcal{S}(\mathbb{H}_\mathcal{I}) \cap \mathcal{S}(\mathbb{H}_\mathcal{J})) = \emptyset, \ \forall \mathcal{I}, \mathcal{J}\in\mathscr{S}(\mathbb{H}), \mathcal{I}\neq \mathcal{J}.
\end{align}
Combined with Lemma~\ref{lem6} and \eqref{eq: 67}, we obtain that, for any $\bar{\mathbf{x}}\in\operatorname{int}(\mathcal{S}(\mathbb{H}_\mathcal{I}))$, the NN-BP solution $\mathbf{x}^\star$ must lie in $\operatorname{cl}(\mathcal{S}(\mathbb{H}_\mathcal{I}))$ rather than in $\operatorname{cl}(\mathcal{S}(\mathbb{H}_\mathcal{J}))$. Hence, $\mathbf{x}^\star$ is unique, as given explicitly by \eqref{eq: opt x express}. Furthermore, combined with Berge's maximum theorem, we can extend \eqref{eq: opt x express} to the boundary of $\mathcal{S}(\mathbb{H}_\mathcal{I})$. We thereby establish the second part of Lemma~\ref{lem 1} under condition \eqref{eq: gamma_{I,j} neq 1}.

Finally, we address the case where condition \eqref{eq: gamma_{I,j} neq 1} does not hold. If there exist $\mathcal{I}\in\mathscr{R}(\mathbb{H})$ and $j\in \mathcal{I}^c$ such that $\gamma_{\mathcal{I},j}=1$, then the uniqueness of $\mathbf{x}^\star$ is no longer guaranteed; see Example~\ref{example3}. To resolve this, Algorithm~\ref{alg 1} can be viewed as introducing a slight perturbation to $\mathbb{H}$. Specifically, let $\delta_1>\cdots>\delta_{n_\textnormal{T}}$ be a sequence of sufficiently small perturbation parameters, and examine all coefficients $\gamma_{\mathcal{I},j}$ with $j=1,\cdots,n_\textnormal{T}$. Upon encountering the first equality $\gamma_{\mathcal{I},j}=1$, the associated vector $\mathbf{h}_j$ is perturbed by $(1+\delta_1)$ to break the tie. If a second tie occurs, then $(1+\delta_2)$ is applied, and the process continues similarly. The proof of Lemma~\ref{lem 1} is completed by taking the limits $\delta_1, \cdots, \delta_{n_\textnormal{T}}\rightarrow 0$
and utilizing Berge's maximum theorem, whose details are omitted here.

\subsection{Proof of Lemma~\ref{lem7}}\label{app: proof of lem7}
We prove Lemma~\ref{lem7} by contradiction.
Suppose that there exist two distinct index sets $\mathcal{I},\mathcal{J}\in\mathscr{S}(\mathbb{H})$ such that
\begin{align}
	\textnormal{Vol}_{n_\textnormal{R}} \left\{ \mathcal{S}_\textnormal{unit}(\mathbb{H}_\mathcal{I}) \cap \mathcal{S}_\textnormal{unit}(\mathbb{H}_\mathcal{J})\right\} \neq 0. \label{eq: S_I and S_J neq 0}
\end{align}
By utilizing the intersection of $\mathcal{S}_\textnormal{unit}(\mathbb{H}_\mathcal{I})$ and $\mathcal{S}_\textnormal{unit}(\mathbb{H}_\mathcal{I})$, we rewrite $\mathcal{I}$ and $\mathcal{J}$ as
\begin{IEEEeqnarray}{rCl}
	\subnumberinglabel{eq: 72}
	\mathcal{I}&=\{\mathcal{I}_\textnormal{int},\mathcal{I}_\textnormal{diff}\}, \quad \mathcal{I}_\textnormal{int}\cap\mathcal{I}_\textnormal{diff}=\emptyset, \\
	\mathcal{J}&=\{\mathcal{J}_\textnormal{int},\mathcal{J}_\textnormal{diff}\},\quad \mathcal{J}_\textnormal{int}\cap\mathcal{J}_\textnormal{diff}=\emptyset,
\end{IEEEeqnarray}
where $\mathcal{I}_\textnormal{int}$ and $\mathcal{J}_\textnormal{int}$
are related to the intersection of $\mathcal{S}_\textnormal{unit}(\mathbb{H}_\mathcal{I})$ and $\mathcal{S}_\textnormal{unit}(\mathbb{H}_\mathcal{J})$, $\mathcal{I}_\textnormal{diff}=\mathcal{I}\setminus\mathcal{I}_\textnormal{int}$, and $\mathcal{J}_\textnormal{diff}=\mathcal{J}\setminus\mathcal{J}_\textnormal{int}$. 
By \eqref{eq: S_I and S_J neq 0}, we can obtain 
\begin{align}
	\mathrm{Vol}_{n_\textnormal{R}} \bigl(\mathcal{S}_\textnormal{unit}(\mathbb{H}_{\{\mathcal{I}_\textnormal{int},\mathcal{J}_\textnormal{int}\}})\bigr)\neq 0,
\end{align}
which indicates that
\begin{align}
	\textnormal{rank}(\mathbb{H}_{\{\mathcal{I}_\textnormal{int},\mathcal{J}_\textnormal{int}\}})=n_\textnormal{R}. \label{eq: 75}
\end{align} 
Thus, there exists $j\in\mathcal{J}_\textnormal{int}$ such that $j\in\mathcal{I}^c$,\footnote{Otherwise, for any $j\in\mathcal{J}_\textnormal{int}$ we have $j\in\mathcal{I}$, i.e., $\mathcal{J}_\textnormal{int}\in\mathcal{I}$. Recall that $\textnormal{rank}(\mathbb{H}_{\{\mathcal{I}_\textnormal{int},\mathcal{J}_\textnormal{int}\}})=n_\textnormal{R}$ and $\textnormal{rank}(\mathbb{H}_\mathcal{I})=n_\textnormal{R}$. It follows that $\{\mathcal{I}_\textnormal{int},\mathcal{J}_\textnormal{int}\}
=
\mathcal{I}$. From this, there are two possibilities:
\begin{itemize}
	\item $\mathcal{S}(\mathbb{H}_{\mathcal{I}})$ is completely contained in $\mathcal{S}(\mathbb{H}_{\mathcal{J}})$. In this case, there exists $k\in\mathcal{J}_\textnormal{diff}$ with $k\in\mathcal{I}^c$ such that 
	$\gamma_{\mathcal{I},k}>1$, which violates \eqref{eq: mathscr{S}(H)};
	\item $\mathcal{S}(\mathbb{H}_{\mathcal{I}})$ in not completely contained in $\mathcal{S}(\mathbb{H}_{\mathcal{J}})$. In this case, there exists $k\in\mathcal{I}$ with $k\in\mathcal{J}^c$ such that 
	$\gamma_{\mathcal{J},k}>1$, which also violates \eqref{eq: mathscr{S}(H)}.
\end{itemize}} and there exists $i\in\mathcal{I}_\textnormal{int}$ such that $i\in\mathcal{J}^c$.

Note that, for any $\mathcal{K}\in\mathscr{S}(\mathbb{H})$ and $k\in\mathcal{K}^c$, we have $\gamma_{\mathcal{K},k}< 1$, which follows directly from \eqref{eq: mathscr{S}(H)}. Thus, we obtain
\begin{align}
	\gamma_{\mathcal{I},j} & < 1,\\
	\gamma_{\mathcal{J},i} & < 1,
\end{align} 
which shows that $\mathbf{h}_i$ lies strictly inside $\mathcal{S}_\textnormal{unit}(\mathbb{H}_\mathcal{J})$, and $\mathbf{h}_j$ lies strictly inside $\mathcal{S}_\textnormal{unit}(\mathbb{H}_\mathcal{I})$. This leads to a contradiction, since it violates the geometric relationship between the two subcones. We hence conclude that the assumption is invalid, finishing the proof of Lemma~\ref{lem7}.


\begin{figure}[h!]
	\centering
	\vspace{-3mm}
	\includegraphics[width=2.0in]{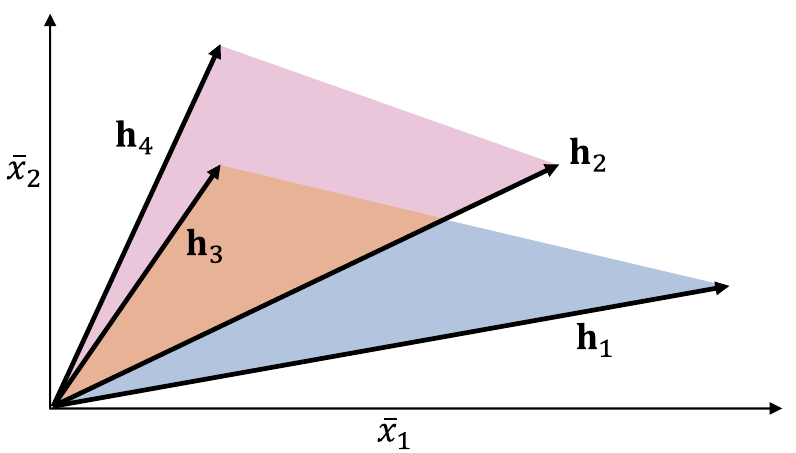}
	\caption{Geometry of $\mathcal{S}_\textnormal{unit}(\mathbb{H}_{\{1,3\}})$ and $\mathcal{S}_\textnormal{unit}(\mathbb{H}_{\{2,4\}})$.}
	\label{fig:ex4}
	\vspace{-1mm}
\end{figure}

We use the following example to provide an intuitive illustration of the above proof.
Consider a $2\times4$ channel $\mathbb{H}=\left(\mathbf{h}_1,\mathbf{h}_2,\mathbf{h}_3,\mathbf{h}_4\right)$, where $\mathscr{S}(\mathbb{H})=\{\mathcal{I}=\{1,3\},\mathcal{J}=\{2,4\}\}$. Assume that
\begin{align}
	\textnormal{Vol}_{2} \left\{ \mathcal{S}_\textnormal{unit}(\mathbb{H}_{\{1,3\}}) \cap \mathcal{S}_\textnormal{unit}(\mathbb{H}_{\{2,4\}})\right\} \neq 0. \label{eq: 77}
\end{align}
Fig.~\ref{fig:ex4} shows the geometries of $\mathcal{S}_\textnormal{unit}(\mathbb{H}_{\{1,3\}})$ and $\mathcal{S}_\textnormal{unit}(\mathbb{H}_{\{2,4\}})$, where the orange shaded region denotes their intersection. From the figure, we obtain $\mathcal{I}_\textnormal{int}=\{3\}$ and $\mathcal{J}_\textnormal{int}=\{2\}$. Then, with the relation in \eqref{eq: 72}, we further obtain that $\mathcal{I}_\textnormal{diff}=\{1\}$ and $\mathcal{J}_\textnormal{diff}=\{4\}$. 

Note that $\{2\}\in\mathcal{J}_\textnormal{int}$ and $\{2\}\in\mathcal{I}^c$, and $\{3\}\in\mathcal{I}_\textnormal{int}$ and $\{3\}\in\mathcal{J}^c$. Combined with the property of $\mathscr{S}(\mathbb{H})$ in \eqref{eq: mathscr{S}(H)}, we have 
\begin{align}
	\mathcal{S}(\mathbb{H}) = \mathcal{S}(\mathbb{H}_{\{1,3\}}) \cup \mathcal{S}(\mathbb{H}_{\{2,4\}}).\label{eq: 78}
\end{align}
To satisfy both \eqref{eq: 77} and \eqref{eq: 78}, it requires that 
\begin{align}	
	\gamma_{\{1,3\},2}& < 1,\\	
	\gamma_{\{2,4\},3}& < 1,
\end{align}
which shows that $\mathbf{h}_2$ lies strictly inside $\mathcal{S}_\textnormal{unit}(\mathbb{H}_{\{1,3\}})$, and $\mathbf{h}_3$ lies strictly inside $\mathcal{S}_\textnormal{unit}(\mathbb{H}_{\{2,4\}})$. This leads to a contradiction.

\section{Proof of Proposition~\ref{prop1}}\label{app: proof of prop}
First, from \eqref{eq: channel model new} and \eqref{eq: bar{x} and x}, we observe that $\mathbf{X}$, $\bar{\mathbf{X}}$, and $\mathbf{Y}$ form a Markov chain. Then, 
\begin{align}
	\const{I}(\mathbf{X};\mathbf{Y}) = 	\const{I}(\bar{\mathbf{X}};\mathbf{Y}), \label{}
\end{align}
which concludes the first part of Proposition~\ref{prop1}. 
Next, the analysis is confined to feasible distributions of $\bar{\mathbf{X}}$. Combined with Lemma~\ref{lem 1}, for any $\bar{\mathbf{x}}\in\mathcal{S}(\mathbb{H}_\mathcal{I})$, the NN-BP solution $\mathbf{x}^\star$ is given by \eqref{eq: opt x express} and the corresponding $\ell_1$-norm is
\begin{align}
	\| \mathbf{x}^\star \|_1 = \|\mathbb{H}_\mathcal{I}^{-1} \bar{\mathbf{x}}\|_1.
\end{align}
Since there is an one-to-one correspondence between $\bar{\mathbf{x}}$ and $\mathbf{x}^\star$, we can restrict the analysis to the input $\mathbf{x}^\star$ without loss of optimality. We can compute that
\begin{align}
	\textsf{E} \{ \| \mathbf{X}^\star \|_1 \} 
	&=\sum\limits_{\mathcal{I}\in\mathscr{S}(\mathsf{
	H})} p_\mathcal{I} \times  \textsf{E} \{ \|\mathbf{X}^\star|\widetilde{\mathcal{I}}=\mathcal{I}\|_1\}\\
	&=\sum\limits_{\mathcal{I}\in\mathscr{S}(\mathsf{
	H})} p_\mathcal{I} \times  \textsf{E} \{ \| \mathbb{H}_\mathcal{I}^{-1} \bar{\mathbf{X}} | \widetilde{\mathcal{I}}=\mathcal{I} \|_1\}\\
	&=\sum\limits_{\mathcal{I}\in\mathscr{S}(\mathsf{
	H})} p_\mathcal{I} \times \| \mathbb{H}_\mathcal{I}^{-1} \textsf{E} \{  \bar{\mathbf{X}} | \widetilde{\mathcal{I}}=\mathcal{I} \}\|_1 \label{eq: 90}\\
	&=\textsf{E}_{\widetilde{\mathcal{I}}}\Bigl\{ \bigl\|\mathbb{H}_{\widetilde{\mathcal{I}}}^{-1} \textsf{E}_{\bar{\mathbf{X}}} \bigl\{ \bar{\mathbf{X}}|\widetilde{\mathcal{I}} \bigr\} \bigr\|_1 \Bigl\},
\end{align}
where \eqref{eq: 90} holds since all elements of
$\mathbb{H}_\mathcal{I}^{-1} \bar{\mathbf{X}} $ are nonnegative. Finally, the second part of Proposition~\ref{prop1}  is a consequence of the fact that $\mathbf{X}^\star$ is a feasible input satisfying \eqref{eq: non-negative} and \eqref{eq: ave cons}.

\section{Derivation of Capacity Results for $n_\textnormal{T}\geq n_\textnormal{R}$}\label{app: proof of case i}

\subsection{Proof of Theorem~\ref{thm: lb}}\label{app: lower bound case i}
The lower bound follows from the mutual information induced by the compound exponential distribution in \eqref{eq: q=p} and \eqref{eq: exp pdf}, i.e.,
\begin{align}
	\const{C} 
	&\geq \const{I}(\bar{\mathbf{X}} ; \mathbf{Y}) \\
	&= \const{h}(\bar{\mathbf{X}}+\mathbf{Z})-\const{h}(\mathbf{Z}) \\
	& \geq \frac{n_\textnormal{R}}{2} \log \left\{\exp \left(\frac{2}{n_\textnormal{R}} \const{h}(\bar{\mathbf{X}})\right)+\exp \left(\frac{2}{n_\textnormal{R}} \const{h}(\mathbf{Z})\right)\right\}\nonumber\\
	&-\const{h}(\mathbf{Z})\qquad \label{eq: EPI}\\
	& \geq \frac{n_\textnormal{R}}{2} \log \left\{1+\frac{\exp \left(\frac{2}{n_\textnormal{R}} \const{h}(\bar{\mathbf{X}})\right)}{2 \pi e \sigma^2}\right\}, \label{eq: C lb com}
\end{align}
where \eqref{eq: EPI} follows from the EPI \cite{Shannon1948}. We further simplify $\const{h}(\bar{\mathbf{X}})$ as
\begin{align}
	&\const{h}(\bar{\mathbf{X}})\nonumber\\
	&= \const{I}(\bar{\mathbf{X}} ; \widetilde{\mathcal{I}})+\const{h}(\bar{\mathbf{X}} \mid \widetilde{\mathcal{I}}) \\
	&= \const{H}(\widetilde{\mathcal{I}})-\const{H}(\widetilde{\mathcal{I}} \mid \bar{\mathbf{X}})+\const{h}(\bar{\mathbf{X}} \mid \widetilde{\mathcal{I}}) \\
	&= \const{H}(\widetilde{\mathcal{I}})+\const{h}(\bar{\mathbf{X}} \mid \widetilde{\mathcal{I}}) \label{eq: H(I,X)=0}\\
	&= \const{H}(\mathbf{p})+\sum_{\mathcal{I} \in \mathscr{S}(\mathbb{H})} p_\mathcal{I} \times \const{h}(\bar{\mathbf{X}} \mid \widetilde{\mathcal{I}}=\mathcal{I}) \\
	&= \const{H}(\mathbf{p})+n_\textnormal{R} \log \left(\frac{\EE}{n_\textnormal{R}}\right)+n_\textnormal{R}\nonumber\\
	&\quad +\,\sum_{\mathcal{I} \in \mathscr{S}(\mathbb{H}) } p_\mathcal{I}\times \log \left(\left|\operatorname{det}\left(\mathbb{H}_\mathcal{I}\right)\right|\right)\label{eq: substitute h(X,I)} \\
	&= -\textsf{D}(\mathbf{p}\| \mathbf{q})+n_\textnormal{R} \log \left(\frac{\EE}{n_\textnormal{R}}\right)+n_\textnormal{R}\nonumber\\
	&\quad+\log \Biggl(\sum_{\mathcal{I} \in \mathscr{S}(\mathbb{H}) } \left|\operatorname{det}\left(\mathbb{H}_\mathcal{I}\right)\right|\Biggr) \label{eq: q def}\\
	&= n_\textnormal{R} \log \left(\frac{\EE}{n_\textnormal{R}}\right)+n_\textnormal{R}+\log \Biggl(\sum_{\mathcal{I} \in \mathscr{S}(\mathbb{H}) } \left|\operatorname{det}\left(\mathbb{H}_\mathcal{I}\right)\right|\Biggr),	\label{eq: h(bar{X}_te)}
\end{align}
where \eqref{eq: H(I,X)=0} follows from $\const{H}(\widetilde{\mathcal{I}}\mid \bar{\mathbf{X}})=0$, \eqref{eq: substitute h(X,I)} follows from $\const{h}(\bar{\mathbf{X}}\mid \widetilde{\mathcal{I}}=\mathcal{I})= n_\textnormal{R} \log \left( \frac{\EE}{n_\textnormal{R}}\right) + n_\textnormal{R} + \log(| \operatorname{det}(\mathbb{H}_\mathcal{I})|)$, and \eqref{eq: h(bar{X}_te)} follows from \eqref{eq: q=p}. Substituting \eqref{eq: h(bar{X}_te)} into \eqref{eq: C lb com}, we can obtain the lower bound in Theorem~\ref{thm: lb}.

\subsection{Proof of Theorem~\ref{thm: ub}}\label{app: upper bound case i}
Let $\bar{\mathbf{X}}^\star$ be the optimal $\bar{\mathbf{X}}$ in \eqref{eq: prop eq1}, with the corresponding $\widetilde{\mathcal{I}}$ and $\mathbf{p}$ denoted by $\widetilde{\mathcal{I}}^\star$ and $\mathbf{p}^\star$ , respectively. Then, an upper bound on channel capacity is given by
\begin{IEEEeqnarray}{rCl}
	\const{C}&
	=& \const{I} \left( \bar{\mathbf{X}}^\star ; \bar{\mathbf{X}}^\star +\mathbf{Z}\right) \\
	&\leq&  \const{I} (\bar{\mathbf{X}}^\star ; \bar{\mathbf{X}}^\star+\mathbf{Z}, \widetilde{\mathcal{I}}^\star ) \\
	&=& \const{I} (\bar{\mathbf{X}}^\star ; \widetilde{\mathcal{I}}^\star)+\const{I}(\bar{\mathbf{X}}^\star ; \bar{\mathbf{X}}^\star+\mathbf{Z} \mid \widetilde{\mathcal{I}}^\star) \\
	&=& \const{H} (\widetilde{\mathcal{I}}^\star)- \const{H} (\widetilde{\mathcal{I}}^\star \mid \bar{\mathbf{X}}^\star)+\const{I}(\bar{\mathbf{X}}^\star ; \bar{\mathbf{X}}^\star+\mathbf{Z} \mid \widetilde{\mathcal{I}}^\star) \\
	&=& \const{H} (\mathbf{p}^\star)+\const{I}(\bar{\mathbf{X}}^\star ; \bar{\mathbf{X}}^\star+\mathbf{Z} \mid \widetilde{\mathcal{I}}^\star)\\
	&=& \const{H} (\mathbf{p}^\star)+\sum\limits_{\mathcal{I}\in\mathscr{S}(\mathbb{H})} p^\star_\mathcal{I}\times \const{I}(\bar{\mathbf{X}}^\star ; \bar{\mathbf{X}}^\star+\mathbf{Z} \mid \widetilde{\mathcal{I}}^\star=\mathcal{I})\qquad\\
	&=& \const{H} (\mathbf{p}^\star)+\sum\limits_{\mathcal{I}\in\mathscr{S}(\mathbb{H})} p^\star_\mathcal{I}\times \const{I}(\mathbf{X}^\star_\mathcal{I}; \mathbf{X}_\mathcal{I}^\star+\mathbf{Z}_\mathcal{I}^\star),\quad\quad \label{eq: C ub 1}
\end{IEEEeqnarray}
where $\mathbf{X}_\mathcal{I}^\star
= \mathbb{H}_\mathcal{I}^{-1} \bar{\mathbf{X}}^\star$, and $\mathbf{Z}_\mathcal{I}^\star
= \mathbb{H}_\mathcal{I}^{-1} \bar{\mathbf{Z}}^\star$. Denote $\mathbf{X}_{\mathcal{I}}^\star=(X_{\mathcal{I},1}^\star,\cdots,X_{\mathcal{I},n_\textnormal{R}}^\star)^\textsf{T}$, and $\mathbf{Z}_{\mathcal{I}}^\star=(Z_{\mathcal{I},1}^\star,\cdots,Z_{\mathcal{I},n_\textnormal{R}}^\star)^\textsf{T}$. For any $\mathcal{I}\in\mathscr{S}(\mathbb{H})$, we define
\begin{align}
	\mathbf{Y}_\mathcal{I}^\star  =\mathbf{X}_\mathcal{I}^\star +\mathbf{Z}_\mathcal{I}^\star,
\end{align}
where $\mathbf{Y}_{\mathcal{I}}^\star=(Y_{\mathcal{I},1}^\star,\cdots,Y_{\mathcal{I},n_\textnormal{R}}^\star)^\textsf{T}$.
On the one hand, the input constraints in \eqref{eq: non-negative} and \eqref{eq: ave cons} imply that
\begin{align}
	&\operatorname{Pr}\{X_{\mathcal{I},l}^\star\in\mathfrak{R}_+\} =1,\quad l\in\{1,\cdots,n_\textnormal{R}\},\\
	&\sum\limits_{\mathcal{I}\in\mathscr{S}(\mathbb{H})} p^\star_\mathcal{I} \sum_{l=1}^{n_\textnormal{R}} \textsf{E}\{X_{\mathcal{I},l}^\star\} \leq \EE. \label{eq: X_I power cons}
\end{align} 
On the other hand, it can be shown that $\mathbf{Z}_\mathcal{I}^\star \sim\mathcal{N}(\mathbf{0}_{n_\textnormal{R}},\sigma^2 \mathbb{H}_{\mathcal{I}}^{-1}\mathbb{H}_{\mathcal{I}}^{-\textsf{T}})$. Since $\mathbf{X}_\mathcal{I}^\star$ and $\mathbf{Z}_\mathcal{I}^\star$ are independent, the channel capacity can be further upper-bounded using the duality-based approach \cite{Moser2004}. For any $\mathcal{I}\in\mathscr{S}(\mathbb{H})$, we choose the auxiliary output distribution of $\mathbf{Y}_\mathcal{I}^\star$ as the following product distribution:
\begin{align}
	R_{\mathbf{Y}_\mathcal{I}^\star}(\mathbf{y}_\mathcal{I}) = \prod\nolimits_{l=1}^{n_\textnormal{R}} R_{Y_{\mathcal{I},l}^\star}(y_{\mathcal{I},l}), \label{eq: R(Y)}
\end{align} 
where $\mathbf{y}_{\mathcal{I}}=(y_{\mathcal{I},1},\cdots,y_{\mathcal{I},n_\textnormal{R}})^\textsf{T}$. Let $W_{\mathcal{I}}(\cdot | \mathbf{X}_\mathcal{I}^\star)$ be the transition law between $\mathbf{X}_\mathcal{I}^\star$ and $\mathbf{Y}_\mathcal{I}^\star$, and let $W_{\mathcal{I},l}(\cdot | X^\star_{\mathcal{I},l})$ be the marginal law between $X_{\mathcal{I},l}^\star$ and $Y_{\mathcal{I},l}^\star$, $l\in\{1,\cdots,n_\textnormal{R}\}$. Then, $\const{I}(\mathbf{X}_\mathcal{I}^\star; \mathbf{X}_\mathcal{I}^\star+\mathbf{Z}_\mathcal{I}^\star)$ can be further upper-bounded as
\begin{IEEEeqnarray}{rCl}
	\IEEEeqnarraymulticol{3}{l}{%
		\const{I}(\mathbf{X}_\mathcal{I}^\star; \mathbf{X}_\mathcal{I}^\star+\mathbf{Z}_\mathcal{I}^\star )
	}\nonumber\\
	& \leq& \textsf{E}_{\mathbf{X}_\mathcal{I}^\star }\Bigl\{\textsf{D}\Bigl(W_\mathcal{I}\left(\mathbf{Y}_\mathcal{I}^\star \mid \mathbf{X}_\mathcal{I}^\star\right) \| R_{\mathbf{Y}_\mathcal{I}^\star}\left(\mathbf{Y}_\mathcal{I}^\star\right)\Bigr)\Bigr\} \\ 
	& =&-\const{h}\left(\mathbf{X}_\mathcal{I}^\star+\mathbf{Z}_\mathcal{I}^\star \mid \mathbf{X}_\mathcal{I}^\star\right)\nonumber\\
	&&-\textsf{E}_{\mathbf{X}_\mathcal{I}^\star}\left\{\textsf{E}_{W_\mathcal{I}\left(\mathbf{Y}_\mathcal{I}^\star \mid \mathbf{X}_\mathcal{I}^\star\right)}\left\{\log \left(R_{\mathbf{Y}_\mathcal{I}^\star}\left(\mathbf{Y}_\mathcal{I}^\star\right)\right)\right\}\right\} \\ 
	& =&-\const{h}\left(\mathbf{X}_\mathcal{I}^\star+\mathbf{Z}_\mathcal{I}^\star \mid \mathbf{X}_\mathcal{I}^\star\right) \nonumber\\
	&& -\sum_{l=1}^{n_\textnormal{R}} \textsf{E}_{X_{\mathcal{I}, l}^\star }\left\{\textsf{E}_{W_{\mathcal{I},l}\left(Y_{\mathcal{I},l}^\star \mid X_{\mathcal{I},l}^\star\right)}\left\{\log \bigl(R_{Y_{\mathcal{I},l}^\star}\left(Y_{\mathcal{I},l}^\star\right)\bigr)\right\}\right\}\qquad \label{eq: susbtite R(y)} \\
	& =&-\frac{n_\textnormal{R}}{2} \log \left(2 \pi e \sigma^2\right)+\log \left(\left|\operatorname{det}\left(\mathbb{H}_\mathcal{I}\right)\right|\right) \nonumber\\
	&& -\sum_{l=1}^{n_\textnormal{R}} \textsf{E}_{X_{\mathcal{I},l}^\star  }\left\{\textsf{E}_{W_{\mathcal{I},l}\left(Y_{\mathcal{I},l}^\star \mid X_{\mathcal{I},l}^\star\right)}\left\{\log \left(R_{Y_{\mathcal{I},l}^\star}\left(Y_{\mathcal{I},l}^\star\right)\right)\right\}\right\},\nonumber\\ \label{eq: 98}	
\end{IEEEeqnarray}
where \eqref{eq: susbtite R(y)} follows from \eqref{eq: R(Y)}. 

First, we choose the output distribution:
\begin{small}
	\begin{IEEEeqnarray}{rCl}		
		&&R_{Y_{\mathcal{I},l}^\star}(y_{\mathcal{I},l})=\nonumber\\
		&&\left\{\begin{array}{ll} 
			\frac{1}{\beta e^{-\frac{\delta^2}{2 \bar{\sigma}_{\mathcal{I},l}^2}}+\sqrt{2 \pi} \bar{\sigma}_{\mathcal{I}, l} Q\left(\frac{\delta}{\bar{\sigma}_{\mathcal{I}, l}}\right)} e^{-\frac{y_{\mathcal{I},l}^2}{2 \bar{\sigma}_{\mathcal{I}, l}^2}}, &y_{\mathcal{I},l} \leq-\delta, \\
			\frac{1}{\beta e^{-\frac{\delta^2}{2 \bar{\sigma}_{\mathcal{I}, l}^2}}+\sqrt{2 \pi} \bar{\sigma}_{\mathcal{I}, l} Q\left(\frac{\delta}{\bar{\sigma}_{\mathcal{I}, l}}\right)} e^{ -\frac{\delta^2}{2 \bar{\sigma}_{\mathcal{I}, l}^2}-\frac{y_{\mathcal{I},l}+\delta}{\beta}}, &y_{\mathcal{I},l}>-\delta,
		\end{array}\right.\qquad \label{eq: R_Y}
	\end{IEEEeqnarray}
\end{small}where $\delta>0$ and $\beta>0$. Then, following steps similar to those in \cite[App. B-F]{Lapidoth2009} and applying the bound $1-Q(\xi)\leq 1$ for $\xi\in\mathfrak{R}$, we obtain
\begin{align}
	& -\textsf{E}_{X_{\mathcal{I},l}^\star }\left\{\textsf{E}_{W_{\mathcal{I},l}\left(Y_{\mathcal{I},l}^\star \mid X_{\mathcal{I},l}^\star \right) } \left\{\log \bigl(R_{Y_{\mathcal{I},l}^\star} \left(Y_{\mathcal{I},l}^\star\right)\bigr)\right\}\right\} \nonumber\\
	& \leq \log \left(\beta e^{-\frac{\delta^2}{2 \bar{\sigma}_{\mathcal{I},l}^2}}+\sqrt{2 \pi} \bar{\sigma}_{\mathcal{I},l} Q\left(\frac{\delta}{\bar{\sigma}_{\mathcal{I},l}}\right)\right)+\frac{1}{2} Q\left(\frac{\delta}{\bar{\sigma}_{\mathcal{I},l}}\right)\nonumber\\
	&\quad+\frac{\delta}{2 \sqrt{2 \pi} \bar{\sigma}_{\mathcal{I},l}} e^{-\frac{\delta^2}{2 \bar{\sigma}_{\mathcal{I}, l}^2}}+\frac{\delta^2}{2 \bar{\sigma}_{\mathcal{I}, l}^2}+\frac{\delta+\textsf{E}\left\{X_{\mathcal{I}, l}^\star\right\}}{\beta}\nonumber\\
	&\quad+\frac{\bar{\sigma}_{\mathcal{I}, l}}{\sqrt{2 \pi} \beta} e^{-\frac{\delta^2}{2 \bar{\sigma}_{\mathcal{I}, l}^2}}. \label{eq: -E_X|I}
\end{align}
Substituting it into \eqref{eq: C ub 1} and \eqref{eq: 98} yields
\begin{align}
	\const{C}
	&\leq 
	\const{H}(\mathbf{p}^\star)-\frac{n_\textnormal{R}}{2} \log \left(2 \pi e \sigma^2\right)  +\sum\limits_{\mathcal{I}\in\mathscr{S}(\mathbb{H})} p^\star_\mathcal{I}	 \log(|\operatorname{det}(\mathbb{H}_\mathcal{I})|)\nonumber\\
	&\quad+\sum\limits_{\mathcal{I}\in\mathscr{S}(\mathbb{H})} p^\star_\mathcal{I} \sum_{l=1}^{n_\textnormal{R}} \Biggl[\log \biggl(\beta e^{-\frac{\delta^2}{2 \bar{\sigma}_{\mathcal{I},l}^2}}+\sqrt{2 \pi} \bar{\sigma}_{\mathcal{I},l} Q\Bigl(\frac{\delta}{\bar{\sigma}_{\mathcal{I},l}}\Bigr)\biggr)\nonumber\\
	&\quad+ \frac{1}{2} Q\left(\frac{\delta}{\bar{\sigma}_{\mathcal{I},l}}\right) +\frac{\delta}{2 \sqrt{2 \pi} \bar{\sigma}_{\mathcal{I},l}}e^{-\frac{\delta^2}{2 \bar{\sigma}_{\mathcal{I}, l}^2}} 
	+ \frac{\delta^2}{2 \bar{\sigma}_{\mathcal{I}, l}^2}
	\nonumber\\
	&\quad+ \frac{\delta+\textsf{E}\bigl\{X_{\mathcal{I}, l}^\star\bigr\}}{\beta}
	+ \frac{\bar{\sigma}_{\mathcal{I}, l}}{\sqrt{2 \pi} \beta} e^{-\frac{\delta^2}{2 \bar{\sigma}_{\mathcal{I}, l}^2}} \Biggr]. \label{eq: 100}	
\end{align}
Combining it with \eqref{eq: X_I power cons}, we can obtain the upper bound \eqref{eq: thm ub1}. 

Second, we choose the output distribution:
\begin{small}
	\begin{IEEEeqnarray}{rCl}		
		R_{Y_{\mathcal{I},l}^\star}(y_{\mathcal{I},l})
		&=&\left\{\begin{array}{ll} 
			\frac{\sqrt{2e}}{\sqrt{2} \nu + \sqrt{\pi e }  \bar{\sigma}_{\mathcal{I}, l} } e^{-\frac{y_{\mathcal{I},l}^2}{2 \bar{\sigma}_{\mathcal{I}, l}^2 } }, &y_{\mathcal{I},l} \leq 0, \\
			\frac{\sqrt{2} }{\sqrt{2}\nu + \sqrt{\pi e }\bar{\sigma}_{\mathcal{I}, l} }
			e^{ -\frac{y_{\mathcal{I},l}}{\nu}}, &y_{\mathcal{I},l}>0,
		\end{array}\right. \qquad \label{}
	\end{IEEEeqnarray}
\end{small}where $\nu> 0$. Then, following steps similar to those in \cite[Sec. II-A]{Jiang2024}, we can obtain
\begin{align}
	& -\textsf{E}_{X_{\mathcal{I},l}^\star }\left\{\textsf{E}_{W_{\mathcal{I},l}\left(Y_{\mathcal{I},l}^\star \mid X_{\mathcal{I},l}^\star\right)} \left\{\log \bigl(R_{Y_{\mathcal{I},l}^\star} \left(Y_{\mathcal{I},l}^\star\right)\bigr)\right\}\right\} \nonumber\\	
	& \leq \frac{\textsf{E} \{ X_{\mathcal{I},l}^\star \} }{\nu} + \frac{\bar{\sigma}_{\mathcal{I}, l}}{\sqrt{2\pi}\nu} + \log\left(\nu+\sqrt{\frac{\pi e}{2}} \bar{\sigma}_{\mathcal{I}, l} \right).
\end{align}
Substituting it into \eqref{eq: C ub 1} and \eqref{eq: 98}, we have
\begin{align}
	\const{C}
	&\leq 
	\const{H}(\mathbf{p}^\star)-\frac{n_\textnormal{R}}{2} \log \left(2 \pi e \sigma^2\right)  +\sum\limits_{\mathcal{I}\in\mathscr{S}(\mathbb{H})} p^\star_\mathcal{I}	 \log(|\operatorname{det}(\mathbb{H}_\mathcal{I})|)\nonumber\\
	&\quad+\sum\limits_{\mathcal{I}\in\mathscr{S}(\mathbb{H})} p^\star_\mathcal{I} \sum_{l=1}^{n_\textnormal{R}} \Biggl[\frac{\textsf{E} \{ X_{\mathcal{I},l}^\star \} }{\nu} + \frac{\bar{\sigma}_{\mathcal{I}, l}}{\sqrt{2\pi}\nu} \nonumber\\
	&\quad+ \log\left(\nu+\sqrt{\frac{\pi e}{2}} \bar{\sigma}_{\mathcal{I}, l} \right) \Biggr].
\end{align}
Combining it with \eqref{eq: X_I power cons}, we can obtain the upper bound \eqref{eq: thm ub2}. 

\subsection{Proof of Theorem~\ref{thm: hsnr}}\label{app: high snr capacity case i}
The achievability follows directly from the lower bound in Theorem~\ref{thm: lb}. Therefore, we mainly focus on proving the converse. For any $\mathcal{I}\in\mathscr{S}(\mathbb{H})$ and $l\in\{1,\cdots,n_\textnormal{R}\}$, we choose $\beta=\beta_{\mathcal{I},l}$ and $\delta=\delta_{\mathcal{I},l}$ in \eqref{eq: R_Y}. Furthermore, we let
\begin{align}
	\beta_{\mathcal{I},l} &= \frac{\EE}{n_\textnormal{R}},\\
	\delta_{\mathcal{I},l} &= \bar{\sigma}_{\mathcal{I},l} \sqrt{\log\left( \frac{\EE}{ \bar{\sigma}_{\mathcal{I},l} }\right)}.	
\end{align}
Substituting them into \eqref{eq: -E_X|I}, we have
\begin{align}
	& -\textsf{E}_{X_{\mathcal{I},l}^\star }\left\{\textsf{E}_{W_{\mathcal{I},l}\left(Y_{\mathcal{I},l}^\star \mid X_{\mathcal{I},l}^\star\right)}\left\{\log \left(R_{Y_{\mathcal{I},l}^\star} \left(Y_{\mathcal{I},l}^\star\right)\right)\right\}\right\} \nonumber\\
	&\leq \log \left(\frac{\EE}{n_\textnormal{R}} e^{-\frac{1}{2} \log \left(\frac{\EE}{\bar{\sigma}_{\mathcal{I}, l}}\right)}+\sqrt{2 \pi} \bar{\sigma}_{\mathcal{I}, l} Q\left(\sqrt{\log \left(\frac{\EE}{\bar{\sigma}_{\mathcal{I}, l}}\right)}\right)\right)\nonumber\\
	&\quad+\frac{1}{2} Q\left(\sqrt{\log \left(\frac{\EE}{\bar{\sigma}_{\mathcal{I}, l}}\right)}\right) \nonumber\\
	&\quad +\frac{1}{2 \sqrt{2 \pi}} \sqrt{\log \left(\frac{\EE}{\bar{\sigma}_{\mathcal{I}, l}}\right)} e^{-\frac{1}{2} \log \left(\frac{\EE}{\bar{\sigma}_{\mathcal{I}, l}}\right)}	\nonumber\\
	&\quad+\frac{1}{2} \log\left(\frac{\EE}{\bar{\sigma}_{\mathcal{I},l} }\right)
	+ \frac{n_\textnormal{R}\bar{\sigma}_{\mathcal{I},l}}{\EE} \sqrt{\log\left(\frac{\EE}{\bar{\sigma}_{\mathcal{I},l}}\right)}	\nonumber\\
	&\quad+\frac{n_\textnormal{R}\textsf{E}\{X_{\mathcal{I},l}^\star\}}{\EE}
	+\frac{n_\textnormal{R}\bar{\sigma}_{\mathcal{I},l}}{\sqrt{2\pi}\EE} 
	e^{-\frac{1}{2} \log\left( \frac{\EE}{\bar{\sigma}_{\mathcal{I},l}}\right)}\\
	&\doteq \log\left(\frac{\EE}{n_\textnormal{R}} \sqrt{\frac{\bar{\sigma}_{\mathcal{I}, l}}{\EE}}\right) 
	+ \frac{1}{2} \log\left(\frac{\EE}{\bar{\sigma}_{\mathcal{I}, l}}\right) 
	+ \frac{n_\textnormal{R}\textsf{E}\{X_{\mathcal{I},l}^\star\}}{\EE}\\
	&= \log(\EE) -\log(n_\textnormal{R}) +  \frac{n_\textnormal{R}\textsf{E}\{X_{\mathcal{I},l}^\star\}}{\EE}.
\end{align}
Combined with \eqref{eq: C ub 1} and \eqref{eq: 98}, we can obtain that at high SNR,
\begin{align}
	\const{C}
	& \leq \const{H}(\mathbf{p}^\star) -\frac{n_\textnormal{R}}{2} \log(2\pi e\sigma^2) 
	+\sum_{\mathcal{I}\in\mathscr{S}(\mathbb{H})} p_\mathcal{I}^\star  \log(|\operatorname{det}(\mathbb{H}_\mathcal{I})|) \nonumber\\
	&\quad + n_\textnormal{R}\log(\EE) - n_\textnormal{R}\log(n_\textnormal{R})  + \frac{n_\textnormal{R}}{\EE} \sum_{\mathcal{I}\in\mathscr{S}(\mathbb{H})} p_\mathcal{I}^\star \sum_{l=1}^{n_\textnormal{R}} \textsf{E}\{X_{\mathcal{I},l}^\star\}\\
	&\leq -\textsf{D}(\mathbf{p}^\star\| \mathbf{q}) + \log\Bigl(\sum_{\mathcal{I}\in\mathscr{S}(\mathbb{H})} |\operatorname{det}(\mathbb{H}_\mathcal{I})|\Bigr) - \frac{n_\textnormal{R}}{2} \log(2\pi e\sigma^2)   \nonumber\\
	&\quad + n_\textnormal{R}\log(\EE) - n_\textnormal{R}\log(n_\textnormal{R}) + n_\textnormal{R}\\
	&\leq \log\Bigl(\sum_{\mathcal{I}\in\mathscr{S}(\mathbb{H})} |\operatorname{det}(\mathbb{H}_\mathcal{I})|\Bigr) - \frac{n_\textnormal{R}}{2} \log(2\pi e\sigma^2) \nonumber\\
	&\quad + n_\textnormal{R}\log(\EE) - n_\textnormal{R}\log(n_\textnormal{R}) + n_\textnormal{R},
\end{align}
where the last inequality holds since $\textsf{D}(\mathbf{p}^\star\| \mathbf{q})\geq 0$. Thus, we complete the converse proof of Theorem~\ref{thm: hsnr}. \footnote{Alternatively, by setting $\mu=\frac{\EE}{n_\textnormal{R}}$ and substituting it into \eqref{eq: thm ub2}, we can also complete the converse proof of Theorem~\ref{thm: hsnr}.}


\section{Derivation of Capacity Results for $n_\textnormal{T}< n_\textnormal{R}$}\label{app: proof of case ii}

Since $\widetilde{\mathbb{H}}$ is invertible and by Lemma~\ref{lem 1}, we have
\begin{align}
	\mathscr{S}(\widetilde{\mathbb{H}})
	=\mathscr{R}(\widetilde{\mathbb{H}})
	=\{1,\cdots,n_\textnormal{T}\}. \label{eq: S(H) equal R(H)}
\end{align} 
Thus, $\mathbf{p}$ and $\mathbf{q}$ reduce to scalars. Note that 
\begin{align}	
	\widetilde{\mathbb{H}}^\textsf{T}\widetilde{\mathbb{H}}	
	&=\mathbb{V}\mathbb{S}_{1:n_\textnormal{T}}^\textsf{T}  \mathbb{S}_{1:n_\textnormal{T}} \mathbb{V}^\textsf{T}\\
	&=\mathbb{V}\mathbb{S}^\textsf{T}  \mathbb{S} \mathbb{V}^\textsf{T}\\
	&=\mathbb{V}\mathbb{S}^\textsf{T} \mathbb{U}^\textsf{T} \mathbb{U} \mathbb{S} \mathbb{V}^\textsf{T}\\
	&=\mathbb{H}^\textsf{T} \mathbb{H}.
\end{align}
Thus, we can obtain
\begin{align}
	|\operatorname{det}(\widetilde{\mathbb{H}})|
	&= |\operatorname{det}(\widetilde{\mathbb{H}}^\textsf{T}\widetilde{\mathbb{H}})|^{\frac{1}{2}}\\
	&=|\operatorname{det}(\mathbb{H}^\textsf{T} \mathbb{H})|^{\frac{1}{2}}, \label{eq: 138}
\end{align}
and
\begin{align}
	\widetilde{\mathbb{H}}^{-1} \widetilde{\mathbb{H}}^{-\textsf{T}}
	&=( \widetilde{\mathbb{H}}^{\textsf{T}} \widetilde{\mathbb{H}})^{-1}\\
	&=( \mathbb{H}^{\textsf{T}} \mathbb{H})^{-1}.\label{eq: 140}
\end{align}
Substituting these, together with $\widetilde{\mathbb{H}}$, into Theorems~\ref{thm: lb}, \ref{thm: ub}, and \ref{thm: hsnr}, we can obtain the results in Theorem~\ref{thm: case II}.


\footnotesize
\begin{thebibliography}{00}
%
%
%
%
\bibitem{Kaushal2017} H. Kaushal and G. Kaddoum, ``Optical communication in space: Challenges and mitigation techniques,'' \emph{IEEE Commun. Surv. Tutor.}, vol. 19, no. 1, pp. 57-96, Firstquarter 2017.

\bibitem{Pathak2015} P. H. Pathak, X. Feng, P. Hu and P. Mohapatra, ``Visible light communication, networking, and sensing: A survey, potential and challenges,'' \emph{IEEE Commun. Surv. Tutor.}, vol. 17, no. 4, pp. 2047-2077, Fourthquarter 2015.

\bibitem{Khalighi2014} M. A. Khalighi and M. Uysal, ``Survey on free space optical communication: A communication theory perspective,'' \emph{IEEE Commun. Surv. Tutor.}, vol. 16, no. 4, pp. 2231-2258, Fourthquarter 2014.


\bibitem{Karunatilaka2015} D. Karunatilaka, F. Zafar, V. Kalavally, and R. Parthiban, ``LED based indoor visible light communications: State of the art,'' \emph{IEEE Commun. Surv. Tutor.}, vol. 17, no. 3, pp. 1649–1678, 3rd Quart., 2015.





	
\bibitem{Zhu2002} X. Zhu and J. M. Kahn, ``Free-space optical communication through atmospheric turbulence channels,'' \emph{IEEE Trans. Commun.}, vol. 50, no. 8, pp. 1293–1300, Aug. 2002.

%

\bibitem{Tavan2012} M. Tavan, E. Agrell, and J. Karout, ``Bandlimited intensity modulation,'' \emph{IEEE Trans. Commun.}, vol. 60, no. 11, pp. 3429–3439, Nov. 2012.

\bibitem{Elzanaty2020} A. Elzanaty and M.-S. Alouini, ``Adaptive coded modulation for IM/DD free-space optical backhauling: A probabilistic shaping approach,'' \emph{IEEE Trans. Commun.}, vol. 68, no. 10, pp. 6388–6402, Oct. 2020.
	
\bibitem{Shannon1949} C. E. Shannon, ``Communication in the presence of noise,'' \emph{Proc. IRE}, vol. 37, no. 1, pp. 10-21, Jan. 1949.

\bibitem{Lapidoth2009} A. Lapidoth, S. M. Moser and M. A. Wigger, ``On the capacity of free-space optical intensity channels,'' \emph{IEEE Trans. Inf. Theory}, vol. 55, no. 10, pp. 4449-4461, Oct. 2009.

\bibitem{Li2025} L. Li, ``Low-SNR asymptotic capacity of two types of optical wireless channels under average-intensity constraints,'' \emph{IEEE Trans. Inf. Theory}, vol. 71, no. 10, pp. 7504-7517, Oct. 2025.
\bibitem{Jiang2024} L. Jiang and L. Li, ``On the capacity of optical wireless channels based on IM-DD,'' \emph{J. Beijing Univ. Post. Telecommun.}, vol. 47, pp. 137-142, 2024.
%
%
%
%
%
%
%
%
%
%
%


\bibitem{Hranilovic2004} S. Hranilovic and F. R. Kschischang, ``Capacity bounds for power- and band-limited optical intensity channels corrupted by Gaussian noise,'' \emph{IEEE Trans. Inf. Theory}, vol. 50, no. 5, pp. 784–795, May 2004. 

\bibitem{Farid2010} A. A. Farid and S. Hranilovic, ``Capacity bounds for wireless optical intensity channels with Gaussian noise,'' \emph{IEEE Trans. Inf. Theory}, vol. 56, no. 12, pp. 6066-6077, Dec. 2010.


\bibitem{Wang2013} J. -B. Wang, Q. -S. Hu, J. Wang, M. Chen and J. -Y. Wang, ``Tight bounds on channel capacity for dimmable visible light communications,'' \emph{J. Light. Technol.}, vol. 31, no. 23, pp. 3771-3779, Dec. 2013.


\bibitem{Chaaban2016} A. Chaaban, J.-M. Morvan, and M.-S. Alouini, ``Free-space optical communications: Capacity bounds, approximations, and a new sphere-packing perspective,'' \emph{IEEE Trans. Commun.}, vol. 64, no. 3, pp. 1176-1191, Mar. 2016. 


\bibitem{Jiang2016} R. Jiang, Z. Wang, Q. Wang and L. Dai, ``A tight upper bound on channel capacity for visible light communications,'' \emph{IEEE Commun. Lett.}, vol. 20, no. 1, pp. 97-100, Jan. 2016.

\bibitem{Moser2017-ISIT} S. M. Moser, M. Mylonakis, L. Wang and M. Wigger, ``Asymptotic capacity results for MIMO wireless optical communication,'' in \emph{Proc. IEEE Int. Symp. Inf. Theory (ISIT)}, Aachen, Germany, 2017, pp. 536-540.

\bibitem{Chaaban2018-hsnr} A. Chaaban, Z. Rezki and M. -S. Alouini, ``Capacity bounds and high-SNR capacity of MIMO intensity-modulation optical channels,'' \emph{IEEE Trans. Wireless Commun.}, vol. 17, no. 5, pp. 3003-3017, May 2018.

\bibitem{Li2020} L. Li, S. M. Moser, L. Wang and M. Wigger, ``On the capacity of MIMO optical wireless channels,'' \emph{IEEE Trans. Inf. Theory}, vol. 66, no. 9, pp. 5660-5682, Sept. 2020.




%
%
%
%
%
%
%
%
%
%
%
%
%
%
%
%
%
%
%
%
%
%
%
%
%
%
%
%
%
%
%
%
%
%
%
%



%
%
%
%
%
%
%
%
\bibitem{Zhou2017} J. Zhou and W. Zhang, ``On the capacity of bandlimited optical intensity channels with Gaussian noise,'' \emph{IEEE Trans. Commun.}, vol. 65, no. 6, pp. 2481–2493, Jun. 2017.

\bibitem{Chaaban2017} A. Chaaban, Z. Rezki, and M.-S. Alouini, ``Fundamental limits of parallel optical wireless channels: Capacity results and outage formulation,'' \emph{IEEE Trans. Commun.}, vol. 65, no. 1, pp. 296–311, Jan. 2017.

%
\bibitem{Cover2006} T. M. Cover and J. A. Thomas, \emph{Elements of Information Theory}, 2nd ed. New York, NY, USA: Wiley, 2006.

\bibitem{Chen1998} S. S. Chen, D. L. Donoho, and M. A. Saunders, ``Atomic decomposition by basis pursuit,'' \emph{SIAM J. Sci. Comput.}, vol. 20, no. 1, pp. 33–61, 1998.
\bibitem{Hyder2016} M. M. Hyder and K. Mahata, ``A sparse recovery method for initial uplink synchronization in OFDMA systems,'' \emph{IEEE Trans. Commun.}, vol. 64, no. 1, pp. 377–386, Jan. 2016.

\bibitem{Fath2013} T. Fath and H. Haas, ``Performance comparison of MIMO techniques for optical wireless communications in indoor environments,'' \emph{IEEE Trans. Commun.}, vol. 61, no. 2, pp. 733-742, Feb. 2013.

\bibitem{Farid2007} A. A. Farid and S. Hranilovic, ``Outage capacity optimization for free-space optical links with pointing errors,'' \emph{J. Lightw. Technol.}, vol. 25, no. 7, pp. 1702-1710, Jul. 2007.

%
%
%
%
%
%
%
%
%
%
%
%
%
%
%
%
%
%
%
%
%
%
%
%
%
%
%
%
%
%
%
%
%
%
%
%
%
%
%
%
%
%
%
%
%
%
%
%
%
%
%
%
%
%
%
%
%
%
%
%
%
%
%
%
%
%
%
%
%
%

\bibitem{Berge1997} C. Berge, \emph{Topological Spaces: Including a Treatment Multi-Valued Functions, Vector Spaces, Convexity}. Chelmsford, MA, USA: Courier Corporation, 1997.


\bibitem{Shannon1948} C. E. Shannon, ``A mathematical theory of communication,'' \emph{Bell System Tech. J.}, vol. 27, no. 3, pp. 379-423, Jul. 1948.
\bibitem{Moser2004} S. M. Moser, ``Duality-based bounds on channel capacity,'' Ph. D. dissertation, ETH Zurich, Zurich, Switzerland, 2004.









\end{thebibliography}
\end{document}